\documentclass{article}
\pdfoutput=1

\usepackage{amsmath}
\usepackage{amsfonts}
\usepackage{amssymb}
\usepackage{amsthm}

\usepackage[colorlinks=true,linkcolor=blue]{hyperref}
\usepackage[numbers]{natbib}
\usepackage[utf8]{inputenc}
\usepackage{graphicx}
\usepackage{url}
\usepackage[utf8]{inputenc}
\usepackage{array,multirow,makecell}
\usepackage[T1]{fontenc}
\usepackage{setspace}
\usepackage{accents}
\usepackage{color}
\usepackage{enumitem}
\usepackage[11pt]{extsizes}
\usepackage[top=2cm, bottom=2cm, left=2cm, right=2cm]{geometry}
\usepackage{systeme}
\usepackage{comment}
\usepackage{bm}
\usepackage{appendix}
\usepackage{authblk}

\theoremstyle{plain}
\newtheorem{theorem}{Theorem}[section]
\newtheorem{Proposition}[theorem]{Proposition}
\newtheorem{definition}[theorem]{Definition}
\newtheorem{lemma}[theorem]{Lemma}
\newtheorem{remark}[theorem]{Remark}

\DeclareGraphicsExtensions{.pdf,.png}

\newcommand{\q}[1]{``#1''}
\DeclareMathOperator*{\esssup}{ess\,sup}
\DeclareMathOperator*{\argmax}{arg\,max}

\DeclareMathOperator*{\arginf}{arg\,inf}
\DeclareMathOperator{\spn}{span}

\usepackage{enumitem,refcount}

\newcounter{countlist}

\makeatletter
\newenvironment{countlist}[2][]
  {\begin{enumerate}[#1]
     \setcounter{countlist}{0}%
     \def\countname{#2}%
     \let\olditem\item
     \renewcommand{\item}{\stepcounter{countlist}\olditem}}
  {  \renewcommand{\@currentlabel}{\arabic{countlist}}%
     \label{\countname}%
   \end{enumerate}}
\makeatother

\def\doubleunderline#1{\underline{\underline{#1}}}

\title{An analysis of linear regression and neural networks approximation for the pricing of swing options}

\author[1,2]{Christian Yeo}
\affil[1]{\footnotesize Sorbonne Université, Laboratoire de Probabilités, Statistique et Modélisation, UMR 8001, Paris, France}
\affil[2]{\footnotesize Engie Global Markets, 92400 Courbevoie, France (e-mail: christian.yeo@sorbonne-universite.fr)}

\date{}
\numberwithin{equation}{section}
\begin{document}

\maketitle

\begin{abstract}
\noindent
Linear regression, firstly introduced for the pricing of American-style options, has since been expanded to include swing options pricing. Swing options price may be viewed as the solution to a Backward Dynamic Programming Principle, which involves a conditional expectation known as the continuation value. The approximation of the continuation value using linear regression involves two levels of approximation. First, the continuation value is replaced by an orthogonal projection over a subspace spanned by a finite set of $m$ squared-integrable functions yielding a first approximation $V^m$ of the swing value function. In this paper, we prove that, with well-chosen regression functions, $V^m$ converges to the swing actual price $V$ as $m \to + \infty$. A similar result is proved when classic regression functions are replaced by neural networks. For both methods (linear regression and neural networks), we analyze the second level of approximation involving practical computation of the swing price using Monte Carlo simulations and yielding an approximation $V^{m, N}$ (where $N$ denotes the Monte Carlo sample size). Especially, we prove that $V^{m, N} \to V^m$ as $N \to + \infty$ for both methods and using a Hilbert basis assumption in the linear regression. Besides, a convergence rate of order $\mathcal{O}\big(\frac{1}{\sqrt{N}} \big)$ is proved in the linear regression case.
\end{abstract}


\vspace{0.3cm}

\section*{Introduction}
Swing contracts \cite{Thompson1995ValuationOP, swingAna} are commonly traded derivatives products in commodity markets and they allow to manage commodity supply. 
These contracts allow their holder to purchase amounts of energy on specific dates (called exercise dates), subject to constraints. The pricing 
\cite{Jaillet2004ValuationOC, BarreraEsteve2006NumericalMF, 10.1007/978-3-642-25746-9_15, Bardou2009OptimalQF, CarmonaTouzi, Kluge} of such a contract is a challenging 
problem that involves finding a vector that represents the amounts of energy purchased through the contract, while maximizing the gained value. This problem is 
doubly-constrained (exercise dates constraint and volume constraints) and its pricing had been addressed using two groups of methods in the literature. One group 
concerns methods that are based on the Backward Dynamic Programming Principle (BDPP) \cite{Bardou2009OptimalQF, BarreraEsteve2006NumericalMF}, which determines the 
swing price backwardly from the expiry of the contract until the pricing date. In BDPP based approach, at each exercise date, the swing value is determined as the 
maximum of the current cash flows plus the continuation value, which is the (conditional) expected value of future cash flows. To compute the continuation value, 
nested simulations may be used, but this can be time-consuming. Alternatively, an orthogonal projection over a vector space spanned by a finite set of squared-integrable 
functions may be used, based on the idea of the linear regression method introduced by Longstaff and Schwartz \cite{Longstaff2001ValuingAO} for the pricing of 
American-style options \cite{Myneni1992ThePO, Parkinson1977OptionPT, Becker2019PricingAH}. Later on, this this mehod had been used to solve more general stochastic control problems 
\cite{Belomestny2009RegressionMF, doi:10.1080/14697688.2015.1088962} and especially in the context of swing contract pricing \cite{BarreraEsteve2006NumericalMF}. 
Despite being widely used by practitioners, in the context of swing pricing, this method has received little studies in terms of convergence. 
The paper \cite{Belomestny2009RegressionMF} analyzes the convergence of general regression methods in the context of stochastic control problems. 
It appears that, for the sake of generality, they made strong assumptions to prove Monte Carlo convergence. Specifically, they imposed a boundedness assumption on 
regression coefficients, and their convergence result is contingent on selecting a particular size of the regression basis based on this boundedness assumption. 
These assumptions are very strong since, even if regression functions are bounded, there is no inherent guarantee that the same holds true for regression coefficients 
without additional assumptions. In this paper, we propose to focus on swing contracts pricing allowing to make less restrictive assumptions. 
Besides, in BDPP based approaches, swing value function depends on cumulative consumption so that to directly apply analysis performed in \cite{Belomestny2009RegressionMF}, 
one may need to include cumulative consumption as a state variable along with the Markov process driving the underlying asset price. However, this can be challenging, 
if not impossible, to implement as it requires to know the joint distribution of the underlying asset price and the cumulative consumption. Indeed, there is no insight on the 
latter joint distribution as it may be noticed in \cite{Bachouch2021DeepNN} where, in the context of storage pricing (contracts whose pricing is closed to that of swing contracts), 
the authors have used uniform sampling for cumulative consumption as a proxy. Since, in practice, we do not have access to the joint distribution of 
the underlying asset price and the cumulative consumption, the analysis of the convergence of BDPP based methods faces what we call \textit{uniform convergence issue}. The latter
just refers to that we will have to uniform convergence result with respect to cumulative consumption.

In this paper, we do not restrict ourselves to linear regression and analyze an alternative method which consist in approximating the continuation value, 
not by an orthogonal projection but, using neural networks. Both methods for approximating the swing contract price are analyzed in a common framework. 
To achieve this, we proceed as in previous works \cite{Lapeyre2019NeuralNR, elfilaliechchafiq:hal-03436046, Clement2002AnAO} by proving some convergence results into two main steps. 
We first replace the continuation value by either an orthogonal projection over a well-chosen basis of regression functions or by neural network. 
We demonstrate that the resulting swing value function, as an approximation of the actual one, converges towards the actual one as the number of functions in the regression basis 
or the number of units per hidden layer in the neural network increases. Furthermore, practically, a Monte Carlo simulation has to be performed. This is needed to compute 
the orthogonal projection coordinates in the linear regression method; which generally has no closed form while it serves as input for training the neural network. 
This leads to a second level of approximation, a Monte Carlo approximation. In this paper, we prove that, under some assumptions, this second approximation converges 
to the first one for both studied methods. Moreover, in the linear regression method, a rate of order $\mathcal{O}(N^{-1/2})$ ($N$ being the size of the Monte Carlo sample) 
of the latter convergence is proved.

\subsubsection*{Contribution of the paper}
\begin{itemize}
    \item In the firm constraints setting, we establish the continuity of the swing value function with respect to cumulative volume. To the best of our knowledge, 
    this result has not been demonstrated previously. This result is necessary, if not indispensable, to circumvent the \emph{uniform convergence issue} mentioned above in the introduction.
    \item For the best of our knowledge, in the swing pricing literature, this paper is the first to analyze theoretical evidence of linear regression (à la Longstaff-Schwartz) practically used in \cite{BarreraEsteve2006NumericalMF}. We additionally analyze a neural network based alternative. Practical evidence of this alternative will be provided in a forthcoming paper.
\end{itemize}

\subsubsection*{Organization of the paper} Section \ref{swing_gen}. provides general background on swing contracts. We thoroughly discuss his pricing and prove one of the main results of this paper concerning the continuity of the swing value function with respect to the cumulative consumption. Section \ref{continuation_approx}. We describe the methodology to approximate the swing value function using either linear regression or neural networks and fix notations and assumptions that will be used in the sequel. Section \ref{cvg_analysis}. We state the main convergence results of this paper as well as some other technical results concerning some deviation inequalities.

\subsection*{Notations}

$\bullet$ $\mathbb{R}^d$ is endowed with the Euclidean norm denoted by $|\cdot|$. $\langle \cdot, \cdot \rangle$ will denote Euclidean inner-product of $\mathbb{R}^d$.

\noindent
$\bullet$ $\mathbb{L}^2_{\mathbb{R}^d}\big(\mathbb{P}\big)$ denotes the space of $\mathbb{R}^d$-valued squared-integrable, with respect to the probability measure $\mathbb{P}$, 
random variables and is equipped with the canonical norm $\| \cdot \|_2$. 

\noindent
$\bullet$ $|\cdot|_{\sup}$ denotes the sup-norm on functional spaces. 

\noindent
$\bullet$ $\mathbb{M}_{d,q}\big(\mathbb{R}\big)$ will represent the space of matrix with $d$ rows, $q$ columns and with real coefficients. When there is no ambiguity, 
we will consider $|\cdot|$ as the Frobenius norm; the space $\mathbb{M}_{d,q}\big(\mathbb{R}\big)$ will be equipped with that norm. For $m \ge 2$, we denote 
by $\mathbb{G}L_m\big(\mathbb{R}\big)$ the subset of $\mathbb{M}_{m,m}\big(\mathbb{R}\big)$ made of non-singular matrices. 

\noindent
$\bullet$ For a metric space $(E, d)$ and a subset $A \subset E$, we define the distance between $x \in E$ and the set $A$ by,
$$d(x, A) = \underset{y \in A}{\inf} \hspace{0.1cm} d(x,y).$$
We denote by $d_H(A, B)$ the Hausdorff metric between two closed, bounded and non-empty sets $A$ and $B$ (equipped with a metric $d$) which is defined by
$$d_H(A, B) = \max\Bigg(\underset{a \in A}{\sup} \hspace{0.1cm}  d(a, B), \hspace{0.2cm} \underset{b \in B}{\sup} \hspace{0.2cm}  d(b, A)\Bigg).$$

\noindent
$\bullet$ Let $E$ be a real pre-Hilbert space equipped with an inner-product $\langle \cdot, \cdot \rangle$ and consider $n$ vectors $x_1, \ldots, x_n$ of $E$. 
The Gram matrix associated to $(x_1, \ldots, x_n)$ is the symmetric non-negative matrix with coefficients $\big(\langle x_i, x_j \rangle\big)_{1 \le i, j \le n}$. 
The determinant of the latter matrix, the Gram determinant, will be denoted by $G(x_1, \ldots, x_n) := \det \big(\langle x_i, x_j \rangle\big)_{1 \le i, j \le n}$.

\section{Swing contract}
\label{swing_gen}
In this first section, we establish the theoretical foundation for swing contracts and their pricing using the \emph{Backward Dynamic Programming Principle (BDPP)}. 
These background being recalled, we will prove important theoretical properties concerning the set of optimal controls involved in the latter principle.

\subsection{Description}
Swing option allows its holder to buy amounts of energy $q_{k}$ at times $t_k$, $k \in \{ 0, \ldots,n-1\}$ (called exercise dates) until the contract maturity $t_n = T$. 
At each exercise date $t_k$, the purchase price (or strike price) is denoted $K_k$ and can be constant (i.e. $K_k := K$ for any $k$) or indexed on a formula. 
In the indexed strike setting, the strike price is calculated as an average of observed commodity prices over a certain period. In this paper, we only consider the fixed strike 
price case. However the indexed strike price case can be treated likewise.

In addition, swing option gives its holder a flexibility on the amount of energy he is allowed to purchase through some (firm) constraints:
\begin{itemize}
    \item \textbf{Local constraints}: at each exercise time $t_k$, the holder of the swing contract has to buy at least $\underline{q}$ and at most $\overline{q}$ i.e.,
    \begin{equation}
        \forall k \in \{0,\ldots,n-1\}, \quad \underline{q}\le q_{k} \le \overline{q}.
    \end{equation}

    \item \textbf{Global constraints}: at maturity, the cumulative purchased volume must be not lower than $\underline{Q}$ and not greater than $\overline{Q}$ i.e.,
    \begin{equation}
    Q_{n} := \sum_{k = 0}^{n-1} q_{k} \in [\underline{Q}, \overline{Q}] \quad \text{with} \quad Q_0 = 0 \quad \text{and} \quad 0 \le \underline{Q} \le \overline{Q} < +\infty.
    \end{equation}
\end{itemize}

At each exercise date $t_k$, the achievable cumulative consumption lies within the following interval, 
\begin{equation}
\label{range_cum_vol}
\mathcal{T}_k :=  \big[Q^{down}(t_k) , Q^{up}(t_k) \big],
\end{equation}

\noindent
where
\begin{equation*}
\left\{
    \begin{array}{ll}
        Q^{down}(t_0) = 0,\\
        \displaystyle Q^{down}(t_k) = \max\big(0, \underline{Q} - (n-k) \cdot \overline{q} \big),\hspace{0.3cm} k \in \{1,\ldots,n-1\}, \\
        Q^{down}(t_n) = \underline{Q}.
    \end{array}
\right.
\end{equation*}
\begin{equation*}
\left\{
    \begin{array}{ll}
        Q^{up}(t_0) = 0,\\
        \displaystyle Q^{up}(t_k) = \min\big(k \cdot \overline{q}, \overline{Q}\big) ,\hspace{0.3cm} k  \in \{1,\ldots,n-1\},\\
        Q^{up}(t_n) = \overline{Q}.
    \end{array}
\right.
\end{equation*}
Note that, in this paper, we only consider \textbf{firm constraints} which means that the holder of the contract cannot violate the constraints. However there exists in the literature alternative settings where the holder can violate the global constraints (not the local ones) but has to pay, at the maturity, a penalty which is proportional to the default; the excess in case of overconsumption and the deficit in case of underconsumption (see \cite{Bardou2009OptimalQF, BarreraEsteve2006NumericalMF}).

The pricing of swing contract is closely related to the resolution of a Stochastic Optimal Control problem (\emph{SOC}), where we aim at finding the optimal decision process 
$(q_k)_{0 \le k \le n-1}$ maximizing the expected value of the discounted cash flows. The latter \emph{SOC} problem is often handled by the \emph{BDPP}.

\subsection{Backward Dynamic Programming Principle}
\label{bdpp_sec}
Let $\left(\Omega, \mathcal{F}, \{ \mathcal{F}_k \}_{0\le k \le n-1}, \mathbb{P} \right)$ be a filtered probability space. We assume that there exists a 
$d$-dimensional (discrete) Markov process $\big(X_{t_k}\big)_{0 \le k \le n-1}$ and a measurable function $g_k : \mathbb{R}^d \to \mathbb{R}$ such that the price 
of the underlying asset of swing contract at time $t_k$ is $S_{t_k} := g_k\big(X_{t_k}\big)$. Throughout this paper, the function $g_k$ will be assumed to have at most linear growth.

The decision process $(q_{k})_{0 \le k \le n-1}$ is defined on the same probability space and is supposed to be $\mathcal{F}_{k}^X$- adapted, where $\mathcal{F}_{k}^X$ is the natural (completed) filtration of $\big(X_{t_k}\big)_{0 \le k \le n-1}$. At each time $t_k$, by purchasing a volume $q_k$, the holder of the contract makes an algebraic profit:
\begin{equation}
\label{payoff_function}
\psi_k\big(q_{k}, X_{t_k} \big) := q_{k} \cdot \big(g_k\big(X_{t_k}\big) - K\big).
\end{equation}
Then for every non-negative $\mathcal{F}_{{k-1}}^X$- measurable random variable $Q_{k}$ (representing the cumulative purchased volume up to time $t_{k-1}$), 
the price of the swing option at time $t_k$ is:
\begin{equation}
    V_k\big(X_{t_k}, Q_{k} \big) = \esssup_{(q_\ell)_{k \le \ell \le n-1} \in \mathcal{A}_{k, Q_k}^{\underline{Q}, \overline{Q}}} \hspace{0.1cm} 
    \mathbb{E}\Bigg[\sum_{\ell=k}^{n-1} e^{-r_\ell(t_\ell - t_k)} \psi_\ell\big(q_\ell, X_{t_\ell} \big) \Big\rvert X_{t_k} \Bigg],
    \label{pricing_swing_formula}
\end{equation}
where the set of admissible decision processes is defined by:
\begin{equation}
    \mathcal{A}_{k, Q}^{\underline{Q}, \overline{Q}} = \Bigg\{(q_\ell)_{k \le \ell \le n-1}, \hspace{0.1cm} q_{\ell} : (\Omega, \mathcal{F}_{\ell}^X, \mathbb{P}) 
    \mapsto [\underline{q}, \overline{q}], \hspace{0.1cm} \sum_{\ell = k}^{n-1} q_{\ell} \in \big[(\underline{Q}-Q)_{+}, \overline{Q}-Q\big] \Bigg\}
\end{equation}
and the expectation in \eqref{pricing_swing_formula} is taken under the risk-neutral probability. $(r_\ell)_{0 \le \ell \le n-1}$ are interest rates over the period 
$[t_0, t_{n-1}]$ that we will assume to be zero. Problem \eqref{pricing_swing_formula} appears to be a constrained stochastic optimal control problem. 
It can be shown (see \cite{Bardou2007WhenAS}) that for all $k \in \{0,\ldots,n-1\}$ and for all $Q_k \in \mathcal{T}_k$, the swing contract price is given by the following backward 
equation, also known as the dynamic programming equation:
\begin{equation}
    \left\{
    \begin{array}{ll}
        V_k(x, Q_k) = \underset{q \in \mathbb{A}_{k}(Q_k)}{\sup} \hspace{0.1cm} \Big[\psi_k(q, x) + \mathbb{E}\big(V_{k+1}( X_{t_{k + 1}}, Q_k + q) \rvert X_{t_k} = x  \big)\Big],\\
        V_{n-1}(x, Q_{n-1}) = \underset{q \in \mathbb{A}_{n-1}(Q_{n-1})}{\sup} \hspace{0.1cm} \psi(q, x),
    \end{array}
\right.
\label{eq_dp_swing}
\end{equation}
where $\mathbb{A}_{k}(Q_k)$ is the set of admissible controls at time $t_k$, with $Q_k$ denoting the cumulative consumption up to time $t_{k-1}$.

Before going any further, it is important to clarify few points. If our objective is the value function, that is $V_k(x, Q_k)$ for any $x \in \mathbb{R}^d$ 
defined in \eqref{eq_dp_swing}, then the set $\mathbb{A}_{k}(Q_k)$ reduces to the following interval,
\begin{equation}
	\label{intervall_adm}
     \mathcal{I}_{k+1}\big(Q_k\big) := \Big[\max\big(\underline{q}, Q^{down}({t_{k+1}}) - Q_k\big), \min\big(\overline{q}, Q^{up}({t_{k+1}}) - Q_k\big) \Big].
\end{equation}
But if our objective is the random variable $V_k\big(X_{t_k}, Q_k\big)$, then for technical convenience, the preceding set $\mathbb{A}_{k}(Q_k)$ is the set of all 
$\mathcal{F}_{k}^X$-adapted processes lying within the interval $\mathcal{I}_{k+1}\big(Q_k\big)$ defined in \eqref{intervall_adm}. A straightforward consequence of the latter is that the optimal control at a given date must not be anticipatory.

It is worth recalling the \textit{bang-bang} feature of swing contracts proved in \cite{Bardou2007WhenAS}. That is, if volume constraints 
$\underline{q}, \overline{q}, \underline{Q}, \overline{Q}$ are whole numbers and $\overline{Q}- \underline{Q}$ is a multiple of $\overline{q} - \underline{q}$, 
then the supremum in \eqref{eq_dp_swing} is attained in one of the boundaries of the interval $\mathcal{I}_{k+1}\big(Q_k\big)$ defined in \eqref{intervall_adm}. 
In this \emph{discrete setting} (integer volume constraints assumption), at each exercise date $t_k$, the set of achievable cumulative consumptions $\mathcal{T}_{k}$ 
defined in \eqref{range_cum_vol} reads,
\begin{equation}
\label{range_cum_vol_discr}
\mathcal{T}_{k} = \mathbb{N} \cap \big[Q^{down}(t_k), Q^{up}(t_k)\big],
\end{equation}
where $Q^{down}(t_k)$ and $Q^{up}(t_k)$ are defined in \eqref{range_cum_vol}.  In this discrete setting, the BDPP \eqref{eq_dp_swing} remains the same. The main difference lies in the fact that, in the discrete setting the supremum involved in the BDPP is in fact a maximum over two possible values enabled by the \textit{bang-bang} feature. From a practical standpoint, this feature allows to drastically reduce the computation time.

Note that this paper aims to study some regression based methods designed to approximate the conditional expectation involved in the BDPP \eqref{eq_dp_swing}. We study two methods which involve linear regression and neural network approximation. In the linear regression, we will go beyond the discrete setting and show that convergence results can be established in general. To achieve this, we need a crucial result which states that the swing value function defined in equation \eqref{eq_dp_swing} is continuous with respect to cumulative consumption. The latter may be established by relying on Berge's maximum theorem (see Proposition \ref{max_th} in Appendix \ref{corresp}). We may justify the use of this theorem through the following proposition, which characterizes the set of admissible volume as a correspondence (we refer the reader to Appendix \ref{corresp} for details on correspondences) mapping attainable cumulative consumption to an admissible control.

\begin{Proposition}
\label{adm_set_continuous}
Denote by $\mathcal{P}\big([\underline{q}, \overline{q}]\big)$ the power set of the interval $[\underline{q}, \overline{q}]$. Then for all $k \in \{0, \ldots,n-1\}$ the correspondence 
\begin{align*}
  \Gamma_k \colon \Big(\mathcal{T}_k, \hspace{0.1cm} |\cdot|\Big) &\to \Big(\mathcal{P}\big([\underline{q}, \overline{q}]\big), \hspace{0.1cm} d_H \Big)\\
  Q &\mapsto \mathbb{A}_{k}(Q)
\end{align*}

\noindent
is continuous and compact-valued.
\end{Proposition}

\begin{proof}
Let $k \in \{0,\ldots,n-1\}$. We need to prove the correspondence $\Gamma_k$ is both lower and upper hemicontinuous. The needed materials about correspondences is given 
in Appendix \ref{corresp} and we rely on the sequential characterization of hemicontinuity. Let us start with the upper hemicontinuity. Since the set 
$[\underline{q}, \overline{q}]$ is compact, then the converse of Proposition \ref{seq_carac_hemicont} in Appendix \ref{corresp} holds true.

\vspace{0.2cm}
Let $Q \in \mathcal{T}_k$ and consider a sequence $(Q_n)_{n \in \mathbb{N}} \in \mathcal{T}_k^{\mathbb{N}}$ which converges to $Q$. Let $(y_n)_{n \in \mathbb{N}}$ be a 
real-valued sequence such that for all $n \in \mathbb{N}, \hspace{0.1cm} y_n$ lies within $\Gamma_k(Q_n)$. Then using the definition of the set of admissible control, 
we know that $\underline{q} \le y_n \le \overline{q}$ yielding $(y_n)_n$ is a real and bounded sequence. Thus, thanks to Bolzano-Weierstrass theorem, there exists a 
subsequence $(y_{\phi(n)})_{n \in \mathbb{N}}$ which is convergent. Let $y = \lim\limits_{n \rightarrow +\infty} y_{\phi(n)} $, then for all $n \in \mathbb{N}$,
\begin{align*}
    y_{\phi(n)} \in \mathbb{A}_{k}(Q_{\phi(n)})& \Longleftrightarrow \max\big(\underline{q}, Q^{down}(t_{k+1}) - Q_{\phi(n)}\big) \le y_{\phi(n)} \le 
    \min\big(\overline{q}, Q^{up}(t_{k+1}) - Q_{\phi(n)}\big).
\end{align*}
Letting $n \to + \infty$ in the preceding inequalities yields $y \in \Gamma_k(Q)$. Which shows that $\Gamma_k$ is upper hemicontinuous at an arbitrary $Q$. Thus the correspondence $\Gamma_k$ is upper hemicontinuous.

\vspace{0.3cm}
For the lower hemicontinuity part, let $Q \in \mathcal{T}_k$, $(Q_n)_{n \in \mathbb{N}} \in \mathcal{T}_k^{\mathbb{N}}$ be a sequence which converges to 
$Q$ and $y \in \Gamma_k(Q)$. Note that if $y = \max(\underline{q}, Q^{down}(t_{k+1}) - Q)$ (or $y = \min(\overline{q}, Q^{up}(t_{k+1}) - Q)$), then it suffices to consider 
$y_n = \max(\underline{q}, Q^{down}(t_{k+1}) - Q_n)$ (or $y_n = \min(\overline{q}, Q^{up}(t_{k+1}) - Q_n)$) so that for all $n \in \mathbb{N}$, $y_n \in \Gamma_k(Q_n)$ and 
$\lim\limits_{n \rightarrow +\infty} y_n = y$.

It remains the case $y \in \mathring{\Gamma}_k(Q)$ (where $\mathring{A}$ denotes the interior of the set $A$). Thanks to Peak point Lemma \footnote{see Theorem 3.4.7 in \url{https://www.geneseo.edu/~aguilar/public/assets/courses/324/real-analysis-cesar-aguilar.pdf} or in \url{https://proofwiki.org/wiki/Peak_Point_Lemma}} one may extract a monotonous subsequence $(Q_{\phi(n)})_n$. Two cases may be distinguished.

\begin{itemize}
    \item $\doubleunderline{(Q_{\phi(n)})_n \hspace{0.1cm} \text{is a non-decreasing sequence}}$.

In this case, for all $n \in \mathbb{N}$, $Q_{\phi(n)} \le Q$. Since $y \in \mathring{\Gamma}_k(Q)$ and $Q \mapsto  \min(\overline{q}, Q^{up}(t_{k+1}) - Q)$ 
is a non-increasing function, it follows $y < \min(\overline{q}, Q^{up}(t_{k+1}) - Q) \le \min(\overline{q}, Q^{up}(t_{k+1}) - Q_{\phi(n)})$ for all $n \in \mathbb{N}$. 
Moreover since
$$y > \lim\limits_{n \rightarrow +\infty} \max(\underline{q}, Q^{down}(t_{k+1}) - Q_{\phi(n)}) \downarrow \max(\underline{q}, Q^{down}(t_{k+1}) - Q),$$
one may deduce that there exists $n_0 \in \mathbb{N}$ such that for all $n \ge n_0, \hspace{0.1cm} y \ge \max(\underline{q}, Q^{down}(t_{k+1}) - Q_{\phi(n)})$. Therefore it suffices to set $y_n = y$ for all $n \ge n_0$ so that $(y_n)_{n \ge n_0}$ is a sequence such that $\lim\limits_{n \rightarrow +\infty} y_n = y$ and $y_n \in \Gamma_k(Q_{\phi(n)})$ for all $n \ge n_0$.

    \item $\doubleunderline{(Q_{\phi(n)})_n \hspace{0.1cm} \text{is a non-increasing sequence}}$.

    Here for all $n \in \mathbb{N}$, we have $Q_{\phi(n)} \ge Q$ so that $y \ge \max(\underline{q}, Q^{down}(t_{k+1})-Q_{\phi(n)})$. 
    Following the proof in the preceding case, one may deduce that there exists $n_0 \in \mathbb{N}$ such that for all 
    $n \ge n_0, \hspace{0.1cm} y \le \min(\overline{q}, Q^{up}(t_{k+1}) - Q_{\phi(n)})$. Thus it suffices to set a sequence $(y_n)_{n \ge n_0}$ identically equal to $y$.
\end{itemize}
This shows that the correspondence $\Gamma_k$ is lower hemicontinuous at an arbitrary $Q$. Thus $\Gamma_k$ is both lower and upper hemicontinous; hence continuous. Moreover, since for all $Q \in \mathcal{T}_k$, $\Gamma_k(Q)$ is a closed and bounded interval in $\mathbb{R}$, then it is compact. This completes the proof.
\end{proof}

In the following proposition, we show the main result of this section concerning the continuity of the value function defined in \eqref{eq_dp_swing} with respect to the cumulative consumption. Let us define the correspondence $C^{*}_k$ by,
\begin{equation}
\label{corresp_set_sol_dp}
C^{*}_k : Q \in \mathcal{T}_k \mapsto \argmax_{q \in \mathbb{A}_{k}(Q)} \hspace{0.1cm}  \Big\{ \psi_k(q, x) + \mathbb{E}\big(V_{k+1}(X_{t_{k + 1}}, Q + q) \rvert X_{t_k} = x\big)\Big\}
\end{equation}
which is the set of solutions of the BDPP \eqref{eq_dp_swing}. Then we have the following proposition.

\begin{Proposition}
\label{cont_val_func}
If for all $k \in \{1,\ldots,n-1\}$ $X_{t_k} \in \mathbb{L}_{\mathbb{R}^d}^1(\mathbb{P})$, then for all $k \in \{0,\ldots,n-1\}$ and all $x \in \mathbb{R}^d$,
\begin{itemize}
	\item The swing value function $Q \in \mathcal{T}_k \mapsto V_k(x, Q)$ is continuous.

	\item The correspondence $C^{*}_k$ (see \eqref{corresp_set_sol_dp}) is non-empty, compact-valued and upper hemicontinuous.
\end{itemize}
\end{Proposition}

\begin{proof}
Let $x \in \mathbb{R}^d$. For technical convenience, we introduce for all $0 \le k \le n-1$ an extended value function $\mathcal{V}_k(x, \cdot)$ defined on the whole real line
\begin{equation*}
    \mathcal{V}_k(x, Q) :=\left\{
    \begin{array}{ll}
         V_k(x, Q) \hspace{2.8cm} \text{if} \hspace{0.2cm} Q \in \mathcal{T}_k = \big[Q^{down}(t_k), Q^{up}(t_k) \big],\\
        V_k(x, Q^{down}(t_k)) \hspace{1.45cm} \text{if} \hspace{0.2cm} Q < Q^{down}(t_k),\\
        V_k(x, Q^{up}(t_k)) \hspace{1.85cm} \text{if} \hspace{0.2cm} Q > Q^{up}(t_k).
    \end{array}
\right.
\end{equation*}
Note that $V_k(x, \cdot)$ is the restriction of $\mathcal{V}_k(x, \cdot)$ on $\mathcal{T}_k$. Propagating continuity over the dynamic programming equation is challenging due to the presence of the variable of interest $Q$ in both the objective function and the domain in which the supremum is taken. To circumvent this issue, we rely on Berge's maximum theorem. More precisely, we use a backward induction on $k$ along with Berge's maximum theorem to propagate continuity through the BDPP.

For any $Q \in \mathcal{T}_{n-1}$, we have $\mathcal{V}_{n-1}(x, Q) = \underset{q \in \mathbb{A}_{n-1}(Q)}{\sup} \psi_{n-1}(q, x)$ and $\psi_{n-1}(\cdot, x)$ (linear in its first argmuent)
is continuous. Thus applying Lemma \ref{conti_sup} yields the continuity of $\mathcal{V}_{n-1}(x, \cdot)$ on $\mathcal{T}_{n-1}$. Moreover, as $\mathcal{V}_{n-1}(x, \cdot)$ is constant outside $\mathcal{T}_{n-1}$ then it is continuous on $(- \infty, Q^{down}(t_{n-1})\big)$ and $\big(Q^{up}(t_{n-1}), +\infty)$. The continuity at $Q^{down}(t_{n-1})$ and $Q^{up}(t_{n-1})$ is straightforward given the construction of $\mathcal{V}_{n-1}$. Thus $\mathcal{V}_{n-1}(x, \cdot)$ is continuous on $\mathbb{R}$. Besides, for all $Q \in \mathbb{R}$,
$$\big|\mathcal{V}_{n-1}(X_{t_{n-1}}, Q)\big| \le \underset{Q \in \mathcal{T}_{n-1}}{\sup} \big|V_{n-1}(X_{t_{n-1}}, Q)\big| \le \overline{q} \cdot \big(|S_{t_{n-1}}| + K \big) \in \mathbb{L}_{\mathbb{R}}^1(\mathbb{P}).$$

\noindent
We now make the following assumption as induction assumption: $\mathcal{V}_{k+1}(x, \cdot)$ is continuous on $\mathbb{R}$ and there exists a real integrable random 
variable $G_{k+1}$ (independent of $Q$) such that, almost surely, $\big|\mathcal{V}_{k+1}(X_{t_{k+1}}, Q)  \big| \le G_{k+1}$. This implies that 
$(q, Q): [\underline{q}, \overline{q}] \times \mathbb{R} \mapsto \psi_k(q, x) + \mathbb{E}\big(\mathcal{V}_{k+1}(X_{t_{k+1}}, Q+q) \rvert X_{t_k} = x  \big)$ is continuous owing to the theorem of continuity under integral sign. Thus owing to Proposition \ref{max_th} one may apply Berge's maximum theorem and we get that $\mathcal{V}_{k}(x, \cdot)$ is continuous on $\mathbb{R}$. In particular $V_k(x, \cdot)$ is continuous on $\mathcal{T}_{k}$ and the correspondence $C_k^{*}$ is non-empty, compact-valued and upper hemicontinuous. This completes the proof.
\end{proof}

As a result of the preceding proposition, one may substitute the $\sup$ in equation \eqref{eq_dp_swing} with a $\max$. This provides another proof for the existence of an optimal consumption an an alternative to the one presented in \cite{Bardou2007WhenAS}. Furthermore, our proof, compared to that in \cite{Bardou2007WhenAS}, does not suppose integer volumes. 

\vspace{0.2cm}
Having addressed the general problem \eqref{eq_dp_swing}, we can now focus on solving it which requires to compute the continuation value.

\section{Approximation of continuation value}
\label{continuation_approx}
The primary challenge in solving the backward equation \eqref{eq_dp_swing} is to compute the continuation value (conditional expectation) involved. A straightforward approach may be to compute this conditional expectation using nested simulations, but this may be time-consuming. Instead, the continuation value may be approximated using either linear regression (as in \cite{BarreraEsteve2006NumericalMF}) or neural networks.

\vspace{0.2cm}
Notice that, it follows from the Markov assumption and the definition of conditional expectation that there exists a measurable function $\Phi_{k+1}^{Q}$ such that
\begin{equation}
\label{approx_cont_val}
\mathbb{E}\big(V_{k+1}(X_{t_{k + 1}}, Q) \rvert X_{t_k}\big) = \Phi_{k+1}^{Q}(X_{t_k}),
\end{equation}
where $\Phi_{k+1}^Q$ solves the following minimization problem,
\begin{equation}
\label{gen_min_pb_func_phi}
\underset{\Phi \in \mathcal{L}^2}{\inf} \hspace{0.1cm} \Big\|\mathbb{E}\big(V_{k+1}(X_{t_{k + 1}}, Q) \rvert X_{t_k}\big) - \Phi\big(X_{t_k} \big) \Big\|_{2}
\end{equation}
and where $\mathcal{L}^2$ denotes the set of all measurable functions that are squared-integrable. Throughout this paper, we use the canonical 
$\mathbb{L}_{\mathbb{R}^d}^r\big(\mathbb{P}\big)$-norm denoted by $\big\|X \big\|_{r}$ and such that
$$\big\|X \big\|_{r}^r = \mathbb{E}^{\mathbb{P}}\big(|X|^r\big) = \int_{\mathbb{R}^d} |x|^r F_X(dx),$$
where $|\cdot|$ denotes the Euclidean norm on $\mathbb{R}^d$ and $F_X$ is the cumulative distribution function of $X$ with respect to the probability measure $\mathbb{P}$.

\vspace{0.2cm}
\color{black}
\noindent
Due to the vastness of $\mathcal{L}^2$, the optimization problem \eqref{gen_min_pb_func_phi} is quite challenging, if not impossible, to solve in practice. It is therefore common to introduce a parameterized form $\Phi_{k+1}(\cdot ; \theta)$ as a solution to the optimization problem \eqref{gen_min_pb_func_phi}. That is, we need to find the appropriate value of $\theta$ in a certain parameter space $\Theta$ such that it solves the following optimization problem:
\begin{equation}
\label{gen_reg_pb}
\underset{\theta \in \Theta}{\inf} \hspace{0.1cm} \Big\|\mathbb{E}\big(V_{k+1}(X_{t_{k + 1}}, Q) \rvert X_{t_k}\big) - \Phi_{k+1}\big(X_{t_k}; \theta \big) \Big\|_{2}.
\end{equation}
Solving the latter problem requires to compute the continuation value whereas it is the target amount. But since the conditional expectation is an orthogonal projection, it follows from Pythagoras' theorem,
\begin{align}
\label{decompo_pytha}
&\Big\|V_{k+1}(X_{t_{k + 1}}, Q) - \Phi_{k+1}(X_{t_k}; \theta) \Big\|_2^2 \nonumber\\
&= \Big\|V_{k+1}(X_{t_{k + 1}}, Q) - \mathbb{E}\big(V_{k+1}( X_{t_{k + 1}}, Q) \rvert X_{t_k}\big)\Big\|_2^2 + 
\Big\|\mathbb{E}\big(V_{k+1}( X_{t_{k + 1}}, Q) \rvert X_{t_k}\big) - \Phi_{k+1}\big(X_{t_k}; \theta \big)  \Big\|_2^2.
\end{align}
Thus any $\theta$ that solves the optimization problem \eqref{gen_reg_pb} also solves the following optimization problem
\begin{equation}
\label{gen_param_reg}
\underset{\theta \in \Theta}{\inf} \hspace{0.1cm} \Big\|V_{k+1}\big(X_{t_{k + 1}}, Q\big) - \Phi_{k+1}\big(X_{t_k}; \theta \big) \Big\|_2.
\end{equation}
Thus in this paper and when needed, we will indistinguishably consider both optimization problems \eqref{gen_reg_pb}, \eqref{gen_param_reg}. In the next section, we discuss the way the function $\Phi_{k+1}(\cdot ; \theta)$ is parameterise depending on whether we use linear regression or neural networks. Moreover, instead of superscript as in \eqref{approx_cont_val}, we adopt the following notation: $\Phi_{k+1}^{Q}(\cdot) := \Phi(\cdot; \theta_{k+1}(Q))$ where $\theta_{k+1}(Q) \in \Theta$ solves the optimization problem  \eqref{gen_reg_pb} or equivalently \eqref{gen_param_reg}. We also dropped the under-script as the function $\Phi$ will be the same for each exercise date; only the parameters $\theta_{k+1}(Q)$ may differ.

\subsection{Linear regression approximation}
\label{lsr_method}
In the linear regression approach, the continuation value is approximated as an orthogonal projection over a subspace spanned by a finite number of squared-integrable functions (see \cite{BarreraEsteve2006NumericalMF}). More precisely, given $m \in \mathbb{N}^{*}$ functions $e^m(\cdot) = \big(e_1(\cdot),\ldots,e_m(\cdot) \big)$, we replace the continuation value involved in \eqref{eq_dp_swing} by an orthogonal projection over the subspace spanned by $e^m\big(X_{t_k} \big)$. This leads to the approximation $V_k^m$ of the actual value function $V_k$ which is defined backwardly as follows,
\begin{equation}
    \left\{
    \begin{array}{ll}
        V^m_k\big(X_{t_k}, Q\big) = \displaystyle \esssup_{q \in \mathbb{A}_{k}(Q)} \hspace{0.1cm} \Big[\psi_k\big(q, X_{t_k}\big) + \Phi_{m}\big(X_{t_k}; \theta_{k+1, m}(Q+q) \big)\Big],\\
        V^m_{n-1}\big(X_{t_{n-1}}, Q\big) = V_{n-1}\big(X_{t_{n-1}}, Q\big) = \displaystyle \esssup_{q \in \mathbb{A}_{n-1}(Q)} \hspace{0.1cm} \psi\big(q, X_{t_{n-1}}\big),
    \end{array}
\right.
\label{estim_orth_proj_dp}
\end{equation}
where $\Phi_{m}$ is defined as follows,
\begin{equation}
\label{approx_ls_proj_orth}
 \Phi_{m}\big(X_{t_k}; \theta_{k+1, m}(Q) \big) = \langle \theta_{k+1, m}(Q), e^m(X_{t_k}) \rangle
\end{equation}
with $\theta_{k+1, m}(Q) \in \Theta_m = \mathbb{R}^m$ being a vector whose components are coordinates of the orthogonal projection and lies within the following set
\begin{equation}
\label{optim_coef_reg_ls}
\mathcal{S}_{k}^{m}(Q) := \arginf_{\theta \in \Theta_m} \hspace{0.1cm} \Big\|V_{k+1}^m\big(X_{t_{k + 1}}, Q\big) - \langle \theta, e^m(X_{t_k}) \rangle \Big\|_2.
\end{equation}
Solving the optimization problem involved in \eqref{optim_coef_reg_ls} leads to a classic linear regression. In this paper, we will assume that $e^m(\cdot)$ forms linearly independent family so that the set $\mathcal{S}_{k}^{m}(Q)$ reduces to a singleton and $\theta_{k+1, m}(Q)$ is uniquely defined  as:
\begin{equation}
\label{coeff_reg_ls}
\theta_{k+1, m}(Q) :=  \big(A_m^k  \big)^{-1} \cdot \mathbb{E}\Big[V_{k+1}^m(X_{t_{k + 1}}, Q)e^m(X_{t_k}) \Big].
\end{equation}
Note that without the latter assumption, $\mathcal{S}_{k}^{m}(Q)$ may not be a singleton. However, in this case, instead of the inverse matrix $\big(A_m^k  \big)^{-1}$, one may consider the Moore–Penrose inverse or pseudo-inverse matrix $\big(A_m^k\big)^\dagger$. In equation \eqref{coeff_reg_ls}, we used the following notation:
$$\mathbb{E} \Big[V_{k+1}^m(X_{t_{k + 1}}, Q)e^m(X_{t_k}) \Big] := \begin{bmatrix}
           \mathbb{E} \big(V_{k+1}^m(X_{t_{k + 1}}, Q)e_1(X_{t_k}) \big) \\
           \mathbb{E} \big(V_{k+1}^m(X_{t_{k + 1}}, Q)e_2(X_{t_k})  \big)\\
           \vdots \\
           \mathbb{E} \big(V_{k+1}^m(X_{t_{k + 1}}, Q)e_m(X_{t_k})  \big)
         \end{bmatrix} \in \mathbb{R}^m,$$
where $A_m^k := \big((A_m^k)_{i, j}  \big)_{1 \le i, j \le m}$ is the (Gram) matrix with entries
\begin{equation}
\label{gram_matrix}
\langle e_i(X_{t_k}), e_j(X_{t_k}) \rangle_{\mathbb{L}^2(\mathbb{P})} := \mathbb{E}\Big[e_i(X_{t_k}) e_j(X_{t_k}) \Big] \quad 1 \le i, j \le m.
\end{equation}
In practice, to compute vector $\theta_{k+1, m}(Q)$, we need to simulate $N$ independent paths $\big(X_{t_0}^{[p]}, \ldots,X_{t_{n-1}}^{[p]}\big)_{1 \le p \le N}$ and use Monte Carlo to evaluate the expectations involved in \eqref{coeff_reg_ls} and \eqref{gram_matrix}. This leads to a second approximation which is a Monte Carlo approximation. For this second approximation, we define the value function $V_k^{m, N}$ starting from the first approximation \eqref{estim_orth_proj_dp} where we replace the expectations by their estimates,
\begin{equation}
    \left\{
    \begin{array}{ll}
        V^{m, N}_k\big(X_{t_{k}}, Q\big) = \displaystyle \esssup_{q \in \mathbb{A}_k(Q)} \hspace{0.1cm} 
        \Big[\psi_k\big(q, X_{t_{k}}\big) + \Phi_{m}\big(X_{t_k}; \theta_{k+1, m, N}(Q+q) \big)\Big],\\
        V^{m, N}_{n-1}\big(X_{t_{n-1}}, Q\big) = V_{n-1}\big(X_{t_{n-1}}, Q\big),
    \end{array}
\right.
\label{estim_second_approx}
\end{equation}
with
\begin{equation}
\theta_{k, m, N}(Q) = \big(A_{m, N}^k  \big)^{-1} \frac{1}{N} \sum_{p=1}^{N} V^{m, N}_{k+1}(X_{t_{k+1}}^{[p]}, Q)e^m(X_{t_k}^{[p]} ),
\end{equation}
using the notation
$$\frac{1}{N} \sum_{p=1}^{N} V^{m, N}_{k+1}(X_{t_{k+1}}^{[p]}, Q)e^m(X_{t_k}^{[p]} ) := \begin{bmatrix}
           \frac{1}{N} \displaystyle \sum_{p=1}^{N} V^{m, N}_{k+1}(X_{t_{k+1}}^{[p]}, Q)e_1(X_{t_k}^{[p]} ) \\
           \frac{1}{N} \displaystyle \sum_{p=1}^{N} V^{m, N}_{k+1}(X_{t_{k+1}}^{[p]}, Q)e_2(X_{t_k}^{[p]} )\\
           \vdots \\
           \frac{1}{N} \displaystyle \sum_{p=1}^{N}  V^{m, N}_{k+1}(X_{t_{k+1}}^{[p]}, Q)e_m(X_{t_k}^{[p]} )
         \end{bmatrix} \in \mathbb{R}^m$$
and $A_{m, N}^k := \big((A_{m, N}^k)_{i, j}  \big)_{1 \le i, j \le m}$ is the (Gram) matrix whose components are
\begin{equation}
\frac{1}{N} \sum_{p=1}^{N} \hspace{0.1cm} e_i\big(X_{t_k}^{[p]} \big) e_j\big(X_{t_k}^{[p]} \big) \hspace{0.6cm} 1 \le i, j \le m.
\end{equation}

This paper investigates a modified version of the linear regression method proposed in \cite{BarreraEsteve2006NumericalMF}. In their approach, the value function at each time step is the result of two steps. First, at each exercise date, they compute the optimal control which is an admissible control that maximizes the value function \eqref{estim_second_approx} along with Monte Carlo simulations. Then, given the optimal control at that date, they compute the value function on this date by summing up all cash-flows from the considered exercise date until the maturity. Recall that we proceed backwardly so that, in practice, it is assumed that at a given exercise date $t_k$, we already have determined optimal control from $t_{k+1}$ to $t_{n-1}$; so that optimal cash flows at these dates may be computed. However, our method directly replaces the continuation value with a linear combination of functions, and the value function is the maximum, over admissible volumes, of the current cash flow plus the latter linear combination of functions. The main difference between both approaches lies in the following. The value function computed in \cite{BarreraEsteve2006NumericalMF} corresponds to \q{actual} realized cash flows whereas the value function in our case does not. However, as recommended in their original paper \cite{Longstaff2001ValuingAO}, after having estimated optimal control backwardly, a forward valuation has to be done in order to eliminate biases. By doing so, our method and that proposed in \cite{BarreraEsteve2006NumericalMF} yield \q{actual} realized cash flows. Thus both approximations meet. 

\vspace{0.2cm}
Our convergence analysis of the linear regression approximation will require some technical assumptions we state below.

\vspace{0.1cm}
\subsubsection*{Main assumptions}

\vspace{0.3cm}
$\bm{\mathcal{H}_1^{LS}}$: For all $k=0,\ldots,n-1$, the sequence $\left(e_i\left( X_{t_k} \right)\right)_{i \ge 1}$ is total in $\mathbb{L}^2\big(\sigma(X_{t_k}) \big)$.

\vspace{0.4cm}
\noindent
$\bm{\mathcal{H}_2^{LS}}$: For all $k=0,\ldots,n-1$, almost surely, $e_0(X_{t_k}),\ldots,e_m(X_{t_k})$ are linearly independent.

\vspace{0.2cm}
\noindent
This assumption ensures the Gram matrix $A_{m}^k$ is non-singular. Moreover, this assumption allows to guarantee the matrix $A_{m, N}^k$ is non-singular for $N$ large enough. Indeed, by the strong law of large numbers, almost surely $A_{m, N}^k \to A_{m}^k \in \mathbb{G}L_m(\mathbb{R})$ (as $N \to +\infty$) with the latter set being an open set.

\vspace{0.4cm}
\noindent
$\bm{\mathcal{H}_{3, r}}$: For all $k = 0, \ldots, n-1$, the random vector $X_{t_k}$ has finite moments at order $r$. $\bm{\mathcal{H}_{3, \infty}}$ will then denote the existence of moments at any order.

\vspace{0.4cm}
\noindent
$\bm{\mathcal{H}_{4, r}^{LS}}$: For all $k = 0, \ldots, n-1$ and for all $j = 1, \ldots, m$, the random variable $e_j (X_{t_k})$ has finite moments at order $r$. $\bm{\mathcal{H}_{4, \infty}^{LS}}$ will then denote the existence of moments at any order.

\vspace{0.2cm}
\noindent
If assumption $\bm{\mathcal{H}_{3, \infty}}$ holds, one may replace assumption $\bm{\mathcal{H}_{4, r}^{LS}}$ by an assumption of linear or polynomial growth of functions $e_j(\cdot)$ with respect to the Euclidean norm.

\vspace{0.2cm}
Before proceeding, note the following comment that will be relevant in the subsequent discussion. Specifically, we would like to remind the reader that the continuity property of the swing actual value function $V_k$ with respect to cumulative consumption, as stated in Proposition \ref{cont_val_func}, also applies to the approximated value function $V_k^m$ involved in the linear regression.

\begin{remark}
\label{cont_high_order}
If we assume that $\bm{\mathcal{H}_{3, 2r}}$ and $\bm{\mathcal{H}_{4, 2r}^{LS}}$ hold true for some $r \ge 1$, then one may show, by a straightforward backward induction, that the functions
$$Q \in \mathcal{T}_{k+1} \mapsto \mathbb{E}\Big(\big|V_{k+1}^m(X_{t_{k+1}}, Q)e^m(X_{t_k})\big|^r  \Big) \hspace{0.4cm} \text{or} \hspace{0.4cm} V_{k+1}^m(X_{t_{k+1}}, Q)$$
are continuous. If only assumption $\bm{\mathcal{H}_{3, r}}$ holds true then $V_{k+1}(X_{t_{k+1}}, \cdot)$ is continuous and there exists a random variable $G_{k+1} \in \mathbb{L}_{\mathbb{R}}^{r}\big(\mathbb{P}\big)$ (independent of $Q$) such that $V_{k+1}(X_{t_{k+1}}, \cdot ) \le G_{k+1}$.
\end{remark}

Instead of using classic functions as regression functions and projecting the swing value function onto the subspace spanned by these regression functions, an alternative approach consists in using neural networks. Motivated by   the function approximation capacity of deep neural networks, as quantified by the \textit{Universal Approximation Theorem} (UAT), our goal is to explore whether a neural network can replace conventional regression functions. In the following section, we introduce a methodology based on neural networks that aims to approximate the continuation value.

\subsection{Neural network approximation}
The goal of a neural network is to approximate complex a function $\Phi : \mathbb{R}^d \to \mathbb{R}^\ell$ by a parametric function $\Phi(\cdot ; \theta)$ where parameters $\theta$ (or weights of the neural network) have to be optimized in a way that the \q{distance} between the two functions $\Phi$ and $\Phi(\cdot; \theta)$ is as small as possible. A neural network can approximate a wide class of complex functions (see \cite{Cybenko1989ApproximationBS, Hornik1991ApproximationCO, Hornik1989MultilayerFN}). A neural network is made of nodes connected to one another where a column of nodes forms a layer (when there are more than one layer in the neural network architecture we speak of a deep neural network). The outermost are the input and output layers and all those in between are called the hidden layers. The connection between the input and output layers through hidden layers is made by means of linear functions and activation functions (non-linear functions). Figure \ref{nn_representation} show an illustration of the architecture of a (deep) neural network.

\begin{figure}[ht]
    \center
    \includegraphics[scale=0.4]{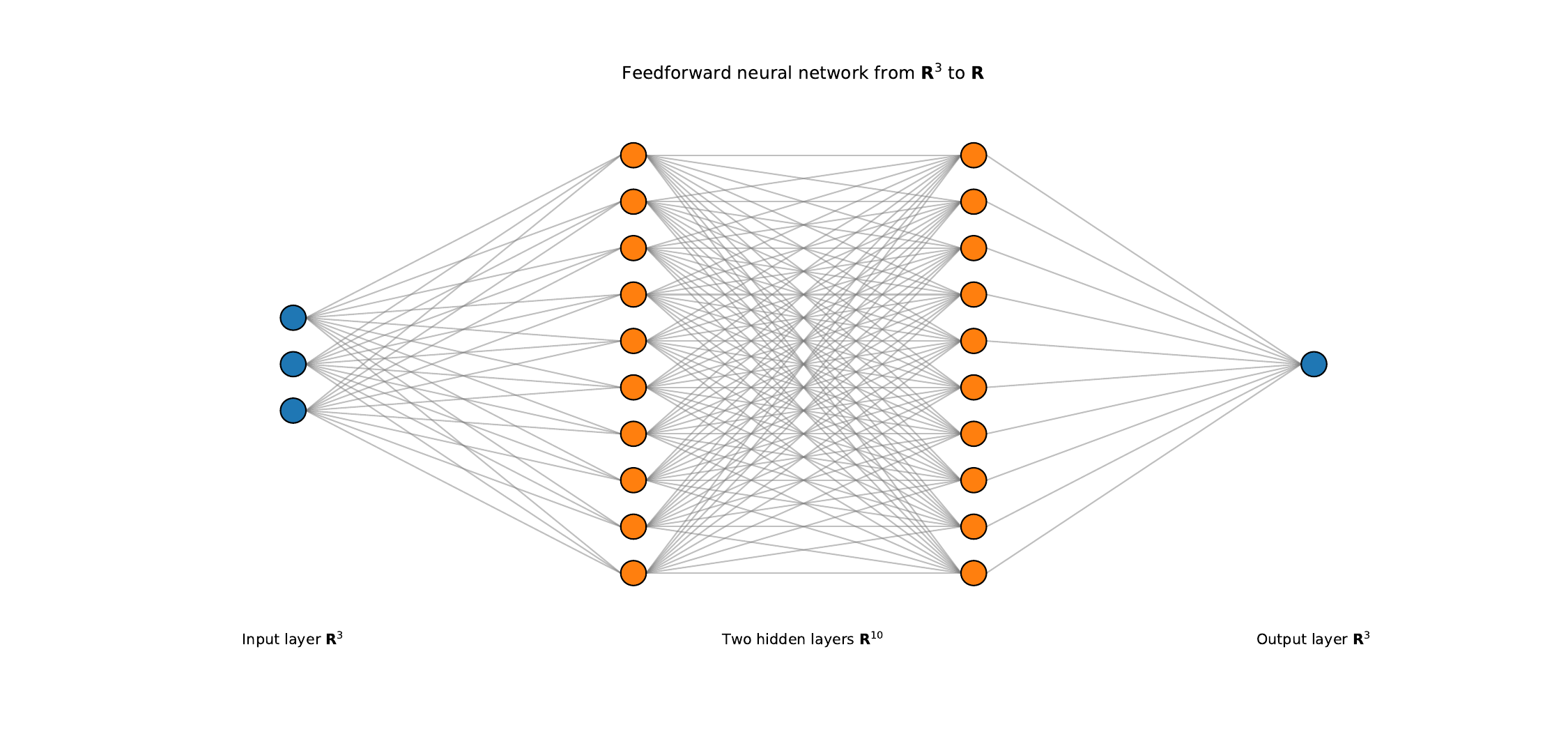}
    \caption{Illustration of (deep) neural network architecture with $d = 3, \ell = 1$}
    \label{nn_representation}
\end{figure}

\noindent
From a mathematical point of view, a neural network can be written as
\begin{equation}
    x \in \mathbb{R}^d \mapsto \Phi(x; \theta) := \psi \circ a_I^{\theta_I} \circ \phi_{q_{I-1}} \circ a_{I-1}^{\theta_{I-1}} \circ \ldots \circ \phi_{q_1} \circ a_1^{\theta_1}(x) \in \mathbb{R}^\ell,
    \label{nn_function_rep}
\end{equation}
where

\vspace{0.2cm}
\noindent
$\rhd$ $I$ is the number of hidden layers representing the depth of the neural network. 

\noindent
$\rhd$ Each layer has weights $\mathcal{W}$ and bias $b$. For all $2 \le i \le I$,

\begin{equation*}
x \in \mathbb{R}^{q_{i-1}} \mapsto a_i^{\theta_i}(x) = \mathcal{W}_i \cdot x + b_i \in \mathbb{R}^{q_i}  \hspace{0.6cm} \text{with} \hspace{0.4cm}  \theta_i = (\mathcal{W}_i, b_i) \in \mathbb{R}^{q_{i-1} \times q_{i}} \times \mathbb{R}^{q_i},
\end{equation*}
and
\begin{equation*}
x \in \mathbb{R}^{d} \mapsto a_1^{\theta_1}(x) = \mathcal{W}_1 \cdot x + b_1 \in \mathbb{R}^{q_1} \hspace{0.6cm} \text{with} \hspace{0.4cm}  \theta_1 = (\mathcal{W}_1, b_1) \in \mathbb{R}^{d \times q_{1}} \times \mathbb{R}^{q_1}.
\end{equation*}

\noindent
$\rhd$ $q_1, \ldots, q_I$ are positive integers denoting the number of nodes per hidden layer and representing the width of the neural network. 

\noindent
$\rhd$ $(\phi_{q_i})_{1 \le i \le I-1}$ are non-linear functions called activation functions and are applied component wise. 

\noindent
$\rhd$ $\psi$ is the activation function for the output layer. 

\vspace{0.2cm}
For the sake of simpler notation, we embed all the parameters of the different layers in a unique high dimensional parameter $\theta = \big(\theta_1,  \ldots, \theta_I \big)  \in \mathbb{R}^{N_q}$ with $N_q = \sum_{i = 1}^{I} q_{i-1} \cdot (1 + q_i)$ (with $q_0 = d$). In order to study neural network approximation, we adopt the same notations as in \cite{Lapeyre2019NeuralNR}. We denote by $\mathcal{NN}_{\infty}$ the set of all neural networks of form \eqref{nn_function_rep}. Then we consider, for some integer $m \ge 1$, $\mathcal{NN}_{m}$ the set of neural networks of form \eqref{nn_function_rep} with at most $m$ nodes per hidden layer and bounded parameters. More precisely, we consider
\begin{equation}
\label{theta_nn}
\Theta_m = \Big\{\mathbb{R}^{d} \times \mathbb{R}^{m} \times \big( \mathbb{R}^{m} \times \mathbb{R}^{m} \big)^{I-2} \times \mathbb{R}^{m} \times  \mathbb{R} \hspace{0.2cm} : \hspace{0.2cm} |\theta| \le \gamma_m  \Big\}
\end{equation}

\noindent
which denotes the set of all parameters (bounded by $\gamma_m$) of a neural network with at most $m$ nodes per hidden layer. $(\gamma_m)_{m \ge 2}$ is an increasing and non-bounded (real) sequence. Thus $\mathcal{NN}_{m}$ is defined as the set of all neural networks which parameters lie in $\Theta_m$,
\begin{equation}
\label{nn_param_m}
\mathcal{NN}_{m} = \big\{\Phi(\cdot; \theta) : \mathbb{R}^d \to \mathbb{R}; \theta \in \Theta_m  \big\}.
\end{equation}

\noindent
Note that $\displaystyle \mathcal{NN}_{\infty} = \bigcup_{m \in \mathbb{N}} \mathcal{NN}_{m}$. In this section, we consider the approximation of the continuation value using neural network. This leads to an approximated value function $V_k^m$ backwardly defined by:
\begin{equation}
\label{dp_approx_nn}
    \left\{
    \begin{array}{ll}
        V_k^m\big(X_{t_k}, Q\big) = \displaystyle \esssup_{q \in \mathbb{A}_k(Q)} \hspace{0.1cm} \Big[\psi_k\big(q, X_{t_k}\big) + \Phi_{m} \big(X_{t_k}; 
        \theta_{k + 1, m}(Q + q)  \big)\Big],\\
        V_{n-1}^m\big(X_{t_{n-1}}, Q\big) = V_{n-1}\big(X_{t_{n-1}}, Q\big),
    \end{array}
\right.
\end{equation}
where $\Phi_{m}(\cdot; \theta)$ denotes a function lying within $\mathcal{NN}_{m}$ with $\theta \in \Theta_m$. Thus $\theta_{k + 1, m}(Q)$ belongs to the following set
\begin{equation}
\label{optim_pb_nn_approx}
\mathcal{S}_{k}^{m}(Q) := \arginf_{\theta \in \Theta_m} \hspace{0.1cm} \Big\|V_{k+1}^m(X_{t_{k+1}}, Q) - \Phi_{m}\big(X_{t_k}; \theta \big) \Big\|_2.
\end{equation}

To analyze the convergence of the neural network approximation, we will rely on their powerful approximation ability. The latter is stated by the \textit{Universal Approximation Theorem}.

\begin{theorem}[Universal Approximation Theorem]
\label{uat_version_l2}
Assume that the activation functions in \eqref{nn_function_rep} are not constant and bounded. Let $\mu$ denote a probability measure on 
$\mathbb{R}^d$, then for any $I \ge 2$, $\mathcal{NN}_{\infty}$ is dense in the space $\mathbb{L}(\mathbb{R}^d, \mu)$ of squared $\mu$-integrable functions.
\end{theorem}

\begin{remark}
\label{th_approx_l2_rq}
As stated in \cite{Lapeyre2019NeuralNR}, Theorem \ref{uat_version_l2} can be seen as follows. For any (real) squared-integrable random variable $Y$ defined on a 
measurable space, there exists a sequence $\displaystyle (\theta_m)_{m \ge 2} \in \prod_{m = 2}^{\infty} \Theta_m$ such that 
$\displaystyle \lim_{p\to\infty} \big\|Y - \Phi_m (X; \theta)  \big\|_2$ for some $\mathbb{R}^d$-valued random vector $X$. Thus, if for all $m \ge 2$, $\theta_m$ solves 
\begin{equation*}
\underset{\theta \in \Theta_m}{\inf} \hspace{0.1cm} \big\|\Phi_m(X; \theta) - Y  \big\|_2,
\end{equation*}
then the sequence $\big(\Phi_m(X; \theta_m) \big)_{m \ge 2}$ converges to $\mathbb{E}(Y \rvert X)$ in $\mathbb{L}^2(\mu)$.
\end{remark}

\vspace{0.2cm}
The universal approximation capacity of neural networks had been widely studied in the literature \cite{Hornik1989MultilayerFN, Hornik1991ApproximationCO, Cybenko1989ApproximationBS}. Some quantitative error bounds have been proved when the function to approximate is sufficiently smooth. A brief overview is presented in the following remark.

\begin{remark}[UAT error bounds]
When the weighted average of the Fourier representation of the function to approximate is bounded, an error bound of the convergence in Remark \ref{th_approx_l2_rq} of order $\mathcal{O}(m^{-1/2})$ had been shown in \cite{baron_approx, baron_approx_2}. It may appears that the dimension of the problem does not degrade the convergence rate but as discussed by the authors, this may be hidden in the Fourier representation. In \cite{ATTALI19971069} it has been proved that, when the activation functions are infinitely continuously differentiable and the function to approximate is $p$-times continuously differentiable and Lipschitz, then the sup-norm of the approximation error on every compact set is bounded by a term of order $\mathcal{O}\big(m^{-(p+1)/d}  \big)$. For a more detailed overview on quantitative error bounds, we refer the reader to \cite{devore2020neural}.
\end{remark}

Note that, as in the linear regression method, in practice, we simulate $N$ independent paths $\big(X_{t_0}^{[p]}, \ldots,X_{t_{n-1}}^{[p]}\big)_{1 \le p \le N}$ and use Monte Carlo approximation to compute the swing value function. For that purpose, we backwardly define the value function $V_k^{m, N}$ by,
\begin{equation}
\label{dp_approx_nn_mc}
    \left\{
    \begin{array}{ll}
        V_k^{m, N}\big(X_{t_k}^{[p]}, Q\big) = \displaystyle \esssup_{q \in \mathbb{A}_k(Q)} \hspace{0.1cm} \Big[\psi_k\big(q, X_{t_k}^{[p]}\big) + 
        \Phi_{m} \big(X_{t_k}^{[p]}; \theta_{k+1, m, N}(Q + q) \big)\Big],\\
        V_{n-1}^{m, N}\big(X_{t_{n-1}}^{[p]}, Q\big) = V_{n-1}\big(X_{t_{n-1}}^{[p]}, Q\big),
    \end{array}
\right.
\end{equation}
where $\theta_{k+1, m, N}(Q)$ lies within the following set,
\begin{equation}
\label{optim_pb_nn_approx_mc}
\mathcal{S}_{k}^{m, N}(Q) := \arginf_{\theta \in \Theta_m} \hspace{0.1cm} \Bigg\{\frac{1}{N} \sum_{p = 1}^{N} \Big|V_{k+1}^{m, N}(X_{t_{k+1}}^{[p]}, Q) - 
\Phi_{m}\big(X_{t_k}^{[p]}; \theta \big) \Big|^2\Bigg\}.
\end{equation}
Sets $\mathcal{S}_{k}^{m}(Q)$ or $\mathcal{S}_{k}^{m, N}(Q)$ (respectively defined in equations \eqref{optim_pb_nn_approx} and \eqref{optim_pb_nn_approx_mc}) generally does not reduces to a singleton. Thus hereafter, the notation $\theta_{k+1, m}(Q)$ or $\theta_{k+1, m, N}(Q)$ will denote an element of the corresponding set $\mathcal{S}_{k}^{m}(Q)$ or $\mathcal{S}_{k}^{m, N}(Q)$.

\section{Convergence analysis}
\label{cvg_analysis}
We conduct a convergence analysis following previous works \cite{Clement2002AnAO, Lapeyre2019NeuralNR, elfilaliechchafiq:hal-03436046}. Our initial focus is to establish a convergence result as the \q{architecture} size used to approximate the continuation value increases. By architecture, we mean either regression functions (in the linear regression) or units per layer (in the neural network). Then, we fix the value of $m$ (representing the architecture's size) and examine the associated Monte Carlo approximation. Let us start with the first step.

\subsection{Convergence with respect to the architecture's size}
We focus on approximations \eqref{estim_orth_proj_dp} and \eqref{dp_approx_nn} of the BDPP \eqref{eq_dp_swing}. In this section, we do not restrict ourselves to the \textit{bang-bang} setting. That is, for both approximation methods, we consider arbitrary volume constraints (not limited to integers).

\subsubsection{Linear regression approximation}
We start by analyzing the first approximation in the linear regression setting \eqref{estim_orth_proj_dp}. We show the convergence of the approximated value function $V_k^m$ to the actual value $V_k$ as $m$ tends to infinity. To state this property, we need the following result.

\begin{Proposition}
\label{prop1}
Let $m$ be a positive integer. Assume $\bm{\mathcal{H}_2^{LS}}$ and $\bm{\mathcal{H}_{3, 2}}$ hold true. Then, for all $k \in \{0,\ldots,n-2\}$, the function
$$Q \mapsto \Big\|\Phi_{m}\big(X_{t_k}; \tilde{\theta}_{k+1, m}(Q)\big) - \mathbb{E}\big(V_{k+1}(X_{t_{k+1}}, Q) \big\rvert X_{t_k}\big) \Big\|_2$$
is continuous on $\mathcal{T}_{k+1}$, where $\Phi_{m}$ is defined in \eqref{approx_ls_proj_orth} and $\tilde{\theta}_{k+1, m}(Q)$ solves the \q{theoretical} optimization problem
\begin{equation}
\label{theore_pb_reg}
\underset{\theta \in \Theta_m}{\inf} \hspace{0.1cm} \Big\|V_{k+1}\big(X_{t_{k + 1}}, Q\big) - \Phi_{m}(X_{t_k}; \theta) \Big\|_2.
\end{equation}
\end{Proposition}

\begin{proof}
Keeping in mind relation \eqref{decompo_pytha}, it suffices to prove that the functions,
\begin{equation}
\label{func_1_cont}
Q \mapsto \Big\|V_{k+1}\big(X_{t_{k+1}}, Q\big) - \mathbb{E}\big(V_{k+1}(X_{t_{k+1}}, Q) \big\rvert X_{t_k}\big) \Big\|_2^2
\end{equation}
and
\begin{equation}
\label{func_2_cont}
Q \mapsto \Big\|V_{k+1}\big(X_{t_{k+1}}, Q\big) - \Phi_{m}\big(X_{t_k}; \tilde{\theta}_{k+1, m}(Q)\big) \Big\|_2^2
\end{equation}
are continuous. Let us start with the first function. Let $Q \in \mathcal{T}_{k+1}$ and consider a sequence $\big(Q_n\big)_n$ which converges to $Q$. We know (as pointed out in Remark \ref{cont_high_order}) that assumption $\bm{\mathcal{H}_{3, 2}}$ entails that $V_{k+1}(X_{t_{k+1}}, \cdot)$ is continuous and there exists $G_{k+1} \in \mathbb{L}_{\mathbb{R}}^2\big(\mathbb{P}\big)$ (independent of $Q$) such that $V_{k+1}(X_{t_{k+1}}, \cdot) \le G_{k+1}$. Thus the Lebesgue dominated convergence theorem implies that,
$$\lim\limits_{n \rightarrow +\infty} \hspace{0.1cm} \big\|V_{k+1}(X_{t_{k+1}}, Q_n) - \mathbb{E}(V_{k+1}(X_{t_{k+1}}, Q_n)\rvert X_{t_k}) \big\|_2^2 = 
\big\|V_{k+1}(X_{t_{k+1}}, Q) - \mathbb{E}(V_{k+1}(X_{t_{k+1}}, Q)\rvert X_{t_k}) \big\|_2^2$$
yielding the continuity of the function defined in \eqref{func_1_cont}. We now prove the continuity of the second function defined in \eqref{func_2_cont}. Using assumption $\bm{\mathcal{H}_2^{LS}}$, it follows from Proposition \ref{gram_det} that,
$$\Big\|\Phi_{m}(X_{t_k}; \tilde{\theta}_{k+1, m}(Q)) - V_{k+1}(X_{t_{k+1}}, Q) \Big\|_2^2 = \frac{G \big(V_{k+1}(X_{t_{k+1}}, Q), e_1(X_{t_k}), \ldots, e_m(X_{t_k}) \big)}{G\big( e_1(X_{t_k}), \ldots, e_m(X_{t_k}) \big)},$$
where $G(x_1, \ldots, x_n)$ denotes the Gram determinant associated to the canonical $\mathbb{L}^2\big(\mathbb{P}\big)$ inner product. Since assumption 
$\bm{\mathcal{H}_{3, 2}}$ entails the continuity of $V_{k+1}(X_{t_{k+1}}, \cdot)$, then owing to the continuity of the determinant, one concludes that 
$Q \in \mathcal{T}_{k+1} \mapsto \Big\|\Phi_{m}(X_{t_k}; \tilde{\theta}_{k+1, m}(Q)) -V_{k+1}(X_{t_{k+1}}, Q) \Big\|_2^2$ is continuous as a composition of continuous functions. 
This completes the proof.
\end{proof}

The preceding proposition allows us to show our first convergence result stated in the following proposition.
\vspace{0.4cm}
\begin{Proposition}
\label{cvg_m_basis}
Under assumptions $\bm{\mathcal{H}_1^{LS}}$, $\bm{\mathcal{H}_{2}^{LS}}$ and $\bm{\mathcal{H}_{3, 2}}$, we have for all $0 \le k \le n-1$,
$$\lim\limits_{m \rightarrow +\infty} \hspace{0.1cm} \underset{Q \in \mathcal{T}_{k}}{\sup} \hspace{0.1cm} \Big\|V^m_k(X_{t_k}, Q) - V_k( X_{t_k}, Q)\Big\|_2 = 0.$$
\end{Proposition}

\begin{proof}
We proceed by a backward induction on $k$. We have, almost surely, $V^m_{n-1}(X_{t_{n-1}}, Q) = V_{n-1}(X_{t_{n-1}}, Q)$ for any $Q \in \mathcal{T}_{n-1}$ 
and therefore the proposition holds true for $k = n-1$. Let us suppose it holds for $k+1$. For all $Q \in \mathcal{T}_{k}$ using the inequality 
$\big|\underset{i \in I}{\sup} \hspace{0.1cm} a_i - \underset{i \in I}{\sup} \hspace{0.1cm} b_i \big| \hspace{0.1cm} \le \hspace{0.1cm} \underset{i \in I}{\sup} \hspace{0.1cm} |a_i-b_i|$, we get,
\begin{align*}
\Big|V^m_k(X_{t_k}, Q) - V_k(X_{t_k}, Q)\Big|^2 &\le \esssup_{q \in \mathbb{A}_k(Q)} \Big|\Phi_{m}\big(X_{t_k}; \theta_{k+1, m}(Q+q)\big) - \mathbb{E}\big(V_{k+1}(X_{t_{k+1}}, Q + q) \big\rvert X_{t_k} \big) \Big|^2.
\end{align*}

\noindent
Taking the expectation in the previous inequality yields,
\begin{align}
\label{pre_bif}
\big\|V^m_k(X_{t_k}, Q) &- V_k(X_{t_k}, Q) \big\|_2^2 \nonumber \\
&\le \mathbb{E}\Bigg(\esssup_{q \in \mathbb{A}_k(Q)} \Big|\Phi_{m}(X_{t_k}; \theta_{k+1, m}(Q+q)) - \mathbb{E}(V_{k+1}(X_{t_{k+1}}, Q + q) \rvert X_{t_k} ) \Big|^2\Bigg).
\end{align}
To interchange the essential supremum with the expectation, we rely on the bifurcation property. For all $q \in \mathbb{A}_k(Q)$, consider
$$A_k^m(Q, q) := \Big|\Phi_{m}(X_{t_k}; \theta_{k+1, m}(Q+q)) - \mathbb{E}\big(V_{k+1}( X_{t_{k+1}}, Q + q) \rvert X_{t_k} \big)\Big|^2.$$
Then for all $q_1, q_2 \in \mathbb{A}_k(Q)$ define the following random variable
\begin{equation}
\label{my_eq_help1}
q_A^{*} = q_1 \cdot \mathrm{1}_{\{A_k^m(Q, q_1) \ge A_k^m(Q, q_2)\}} + q_2 \cdot \mathrm{1}_{\{A_k^m(Q, q_1) < A_k^m(Q, q_2)\}}.
\end{equation}
It follows from the definition of $\Phi_{m}$ in \eqref{approx_ls_proj_orth} and that of the conditional expectation that $A_k^m(Q, q)$ is 
$\sigma\left(X_{t_k} \right)$-measurable for all $q \in \mathbb{A}_k(Q)$. Thus using \eqref{my_eq_help1} yields 
$q_A^{*} \in \mathbb{A}_k(Q)$ and $A_k^m(Q, q_A^{*}) = \max\left(A_k^m(Q, q_1), A_k^m(Q, q_2) \right)$. Therefore one may use the bifurcation property in \eqref{pre_bif} and we get,
    \begin{align}
    \label{eq_sup_diff_val_func}
        \big\|V^m_k(X_{t_k}, Q) &- V_k(X_{t_k}, Q) \big\|_2^2 \nonumber\\
        &\le \underset{q \in \mathbb{A}_k(Q)}{\sup} \big\|\Phi_{m}(X_{t_k}; \theta_{k+1, m}(Q+q)) - 
        \mathbb{E}(V_{k+1}(X_{t_{k+1}}, Q + q) \rvert X_{t_k} ) \big\|_2^2 \nonumber\\
        &\le 2 \underset{q \in \mathbb{A}_k(Q)}{\sup} \big\|\Phi_{m}(X_{t_k}; \theta_{k+1, m}(Q+q)) - \Phi_{m}(X_{t_k}; \tilde{\theta}_{k+1, m}(Q+q)) \big\|_2^2  \nonumber\\
        & \quad +2\underset{q \in \mathbb{A}_k(Q)}{\sup} \big\|\Phi_{m}(X_{t_k}; \tilde{\theta}_{k+1, m}(Q+q)) - \mathbb{E}(V_{k+1}( X_{t_{k+1}}, Q + q) \rvert X_{t_k} ) \big\|_2^2,
    \end{align}
where in the last inequality, we used Minkowski inequality. $\tilde{\theta}_{k+1, m}(Q+q)$ solves the \q{theoretical} optimization problem \eqref{theore_pb_reg}. Note that in the latter problem, we introduced the actual (not known) value function $V_{k+1}$ unlike in equation \eqref{optim_coef_reg_ls}. This is just a theoretical tool as the preceding optimization problem cannot be solved since we do not know the actual value function $V_{k+1}$. Thus taking the supremum in \eqref{eq_sup_diff_val_func} yields,
    \begin{align}
    \label{breakpoint}
        \underset{Q \in \mathcal{T}_{k}}{\sup} \big\|V^m_k(X_{t_k}, Q) &- V_k(X_{t_k}, Q) \big\|_2^2 \nonumber \\
         &\le 2 \underset{Q \in \mathcal{T}_{k+1}}{\sup} \big\|\Phi_{m}(X_{t_k}; \theta_{k+1, m}(Q)) - \Phi_{m}(X_{t_k}; \tilde{\theta}_{k+1, m}(Q))  \big\|_2^2  \nonumber \\
        &\quad +2\underset{Q \in \mathcal{T}_{k+1}}{\sup} \big\|\Phi_{m}(X_{t_k}; \tilde{\theta}_{k+1, m}(Q)) - \mathbb{E}(V_{k+1}(X_{t_{k+1}}, Q) \rvert X_{t_k}) \big\|_2^2,
    \end{align}
where we used the fact that, for all $Q \in \mathcal{T}_{k}$ and all $q \in \mathbb{A}_k(Q)$ we have $Q + q \in \mathcal{T}_{k+1}$. Besides, recall that $\Phi_{m}(X_{t_k}; \tilde{\theta}_{k+1, m}(Q))$ and $\Phi_{m}(X_{t_k}; \theta_{k+1, m}(Q))$ are orthogonal projections of $V_{k+1}(X_{t_{k+1}}, Q)$ and $V_{k+1}^m(X_{t_{k+1}}, Q)$ on the subspace spanned by $e^m(X_{t_k})$. Then knowing that the orthogonal projection is 1-Lipschitz, we have
$$\underset{Q \in \mathcal{T}_{k+1}}{\sup} \big\|\Phi_{m}(X_{t_k}; \theta_{k+1, m}(Q)) - \Phi_{m}(X_{t_k}; \tilde{\theta}_{k+1, m}(Q))\big\|_2^2 \le 
\underset{Q \in \mathcal{T}_{k+1}}{\sup} \big\|V^m_{k+1}(X_{t_{k+1}}, Q) - V_{k+1}(X_{t_{k+1}}, Q)\big\|_2^2.$$
Thanks to the induction assumption, the right hand side of the last inequality converges to $0$ as $m \to + \infty$, so that,
\begin{equation}
\label{ind_assump}
\underset{Q \in \mathcal{T}_{k+1}}{\sup} \big\|\Phi_{m}(X_{t_k}; \theta_{k+1, m}(Q)) - \Phi_{m}(X_{t_k}; \tilde{\theta}_{k+1, m}(Q))\big\|_2^2 \hspace{0.2cm} \xrightarrow[m \to + \infty]{} 0.
\end{equation}
It remains to prove that
\begin{equation}
\label{second_lim}
\lim\limits_{m \rightarrow +\infty} \hspace{0.1cm} \underset{Q \in \mathcal{T}_{k+1}}{\sup} \Big\|\Phi_{m}(X_{t_k}; \tilde{\theta}_{k+1, m}(Q)) - 
\mathbb{E}(V_{k+1}( X_{t_{k+1}}, Q) \rvert X_{t_k}) \Big\|_2^2 = 0.
\end{equation}
To achieve, this we rely on Dini's lemma whose assumptions hold true owing to the three following facts.

\subsubsection*{Pointwise convergence}
\noindent
It follows from assumption $\bm{\mathcal{H}_1^{LS}}$ that, for any $Q \in \mathcal{T}_{k+1}$,
$$\lim\limits_{m \rightarrow +\infty}  \Big\|\Phi_{m}(X_{t_k}; \tilde{\theta}_{k+1, m}(Q)) - \mathbb{E}(V_{k+1}( X_{t_{k+1}}, Q) \rvert X_{t_k}) \Big\|_2^2 = 0.$$

\subsubsection*{Continuity}
\noindent
The continuity of $Q \mapsto\Big\|\Phi_{m}(X_{t_k}; \tilde{\theta}_{k+1, m}(Q)) - \mathbb{E}(V_{k+1}( X_{t_{k+1}}, Q) \rvert X_{t_k}) \Big\|_2^2$ is given by Proposition \ref{prop1} under assumptions $\bm{\mathcal{H}_{2}^{LS}}$ and $\bm{\mathcal{H}_{3, 2}}$.

\subsubsection*{Monotony}
\noindent
Denote by $F_m^k := \spn \big( e_1(X_{t_k}), \ldots, e_m(X_{t_k}) \big)$. Then it is straightforward that for any $m \ge 1$, $F_m^k \subseteq F_{m+1}^k$. So that,
\begin{align*}
\Big\|\Phi_{m}(X_{t_k}; \tilde{\theta}_{k+1, m}(Q)) - \mathbb{E}(V_{k+1}( X_{t_{k+1}}, Q) \rvert X_{t_k}) \Big\|_2^2 &= \underset{Y \in F_m^k}{\inf} \hspace{0.1cm} \Big\|\mathbb{E}\big(V_{k+1}( X_{t_{k+1}}, Q) \rvert X_{t_k}\big) - Y\Big|\Big|_2^2\\
&\ge \underset{Y \in F_{m+1}^k}{\inf} \hspace{0.1cm} \Big\|\mathbb{E}\big(V_{k+1}( X_{t_{k+1}}, Q) \rvert X_{t_k}\big) - Y\Big\|_2^2\\
&= \Big\|\Phi_{m+1}(X_{t_k}; \tilde{\theta}_{k+1, m+1}(Q)) - \mathbb{E}(V_{k+1}( X_{t_{k+1}}, Q) \rvert X_{t_k}) \Big\|_2^2.
\end{align*}
Thus the sequence,
$$\left(\Big\|\Phi_{m}(X_{t_k}; \tilde{\theta}_{k+1, m}(Q)) - \mathbb{E}(V_{k+1}( X_{t_{k+1}}, Q) \rvert X_{t_k}) \Big\|_2^2 \right)_{m \ge 1}$$
is non-increasing. From the three preceding properties, one may apply Dini lemma yielding the desired result \eqref{second_lim}. Finally, combining \eqref{ind_assump} and \eqref{second_lim} in \eqref{breakpoint} yields,
$$\lim\limits_{m \rightarrow +\infty} \hspace{0.1cm} \underset{Q \in \mathcal{T}_{k}}{\sup} \big\|V^m_k(X_{t_k}, Q) - V_k(X_{t_k}, Q) \big\|_2^2 = 0.$$
This completes the proof.
\end{proof}

\subsubsection{Neural network approximation}
We now consider the approximation of the continuation value using neural network. We prove a similar result as in Proposition \ref{cvg_m_basis}, when the number of units per hidden layer increases. To this end, we set the following assumptions.

\vspace{0.3cm}
\noindent
$\bm{\mathcal{H}_1^{\mathcal{NN}}}$: For every $m \ge 2$, there exists $q \ge 1$ such that for every $\theta \in \Theta_m$, $\Phi_m(\cdot; \theta)$ has $q$-polynomial growth uniformly in $\theta$.

\vspace{0.4cm}
\noindent
$\bm{\mathcal{H}_2^{\mathcal{NN}}}$: For any $0 \le k \le n-1$, a.s. the random functions $\theta \in \Theta_m \mapsto \Phi_m\big(X_{t_k}; \theta \big)$ are continuous. Owing to the Heine theorem, the compactness of $\Theta_m$ yields the uniform continuity.

\vspace{0.3cm}

\begin{Proposition}
\label{cvg_m_nn}
Assume $\bm{\mathcal{H}_1^{\mathcal{NN}}}$, $\bm{\mathcal{H}_2^{\mathcal{NN}}}$ and $\bm{\mathcal{H}_{3, 2q}}$ (with $q$ involved in assumption $\bm{\mathcal{H}_1^{\mathcal{NN}}}$) hold true. Then, for all $0 \le k \le n-1$,
$$\lim\limits_{m \rightarrow +\infty} \hspace{0.1cm} \underset{Q \in \mathcal{T}_{k}}{\sup} \hspace{0.1cm} \Big\|V^m_k(X_{t_k}, Q) - V_k( X_{t_k}, Q)\Big\|_2 = 0.$$
\end{Proposition}

\begin{proof}
We proceed by a backward induction on $k$. For $k = n-1$, we have almost surely $V^m_{n-1}(X_{t_{n-1}}, Q) = V_{n-1}(X_{t_{n-1}}, Q)$ and therefore the proposition holds true. Let us suppose it holds for $k+1$. In light of the beginning of the proof of Proposition \ref{cvg_m_basis}, we have for all $Q \in \mathcal{T}_{k}$ using the inequality: $|\underset{i \in I}{\sup} \hspace{0.1cm} a_i - \underset{i \in I}{\sup} \hspace{0.1cm} b_i| \hspace{0.1cm} \le \hspace{0.1cm} \underset{i \in I}{\sup} \hspace{0.1cm} |a_i-b_i|$ and triangle inequality,
    \begin{align}
    \label{eq_help_nn_m}
        \big\|V^m_k(X_{t_k}, Q) &- V_k(X_{t_k}, Q) \big\|_2^2 \nonumber\\
        &\le \mathbb{E}\left(\esssup_{q \in \mathbb{A}_k(Q)} \Big|\Phi_{m}\big(X_{t_k}; \theta_{k+1, m}(Q+q)\big) - \mathbb{E}\big(V_{k+1}( X_{t_{k+1}}, Q + q) 
        \rvert X_{t_k} \big) \Big|^2 \right).
    \end{align}
Then we aim to apply the bifurcation property. For all $q \in \mathbb{A}_k(Q)$, consider,
$$A_k^m(Q, q) = \Big|\Phi_{m}\big(X_{t_k}; \theta_{k+1, m}(Q+q)\big) - \mathbb{E}\big(V_{k+1}( X_{t_{k+1}}, Q + q) \rvert X_{t_k} \big)\Big|^2.$$
For all $q_1, q_2 \in \mathbb{A}_k(Q)$ define
$$q_A^{*} = q_1 \cdot \mathrm{1}_{\{A_k^m(Q, q_1) \ge A_k^m(Q, q_2)\}} + q_2 \cdot \mathrm{1}_{\{A_k^m(Q, q_1) < A_k^m(Q, q_2)\}}.$$
Using the definition of the conditional expectation and since activation functions are continuous (assumption $\bm{\mathcal{H}_2^{\mathcal{NN}}}$), 
$A_k^m(Q, q)$ is $\sigma\left(X_{t_k} \right)$-measurable for all $q \in \mathbb{A}_k(Q)$. Moreover, $q_A^{*} \in \mathbb{A}_k(Q)$ and 
$A_k^m(Q, q_A^{*}) = \max\left(A_k^m(Q, q_1), A_k^m(Q, q_2) \right)$. Thus using the bifurcation property and taking the supremum  in \eqref{eq_help_nn_m} yields,
$$\underset{Q \in \mathcal{T}_{k}}{\sup} \big\|V^m_k(X_{t_k}, Q) - V_k(X_{t_k}, Q) \big\|_2^2 \le 
\underset{Q \in \mathcal{T}_{k+1}}{\sup} \Big\|\Phi_{m}\big(X_{t_k}; \theta_{k+1, m}(Q)\big) - \mathbb{E}\big(V_{k+1}( X_{t_{k+1}}, Q) \rvert X_{t_k} \big)\Big\|_2^2.$$
Using Minkowski inequality and the inequality: $(a+b)^2 \le 2(a^2+b^2)$ yields,
    \begin{align*}
        \underset{Q \in \mathcal{T}_{k}}{\sup} \big\|V^m_k(X_{t_k}, Q) - V_k(X_{t_k}, Q) \big\|_2^2 &\le 2 \underset{Q \in \mathcal{T}_{k+1}}{\sup} 
        \big\|\mathbb{E}\big(V_{k+1}^m( X_{t_{k+1}}, Q) \rvert X_{t_k} \big) - \mathbb{E}\big(V_{k+1}( X_{t_{k+1}}, Q) \rvert X_{t_k} \big)\big\|_2^2 \\
        &\quad +2\underset{Q \in \mathcal{T}_{k+1}}{\sup} \big\|\Phi_{m}\big(X_{t_k}; \theta_{k+1, m}(Q)\big) - \mathbb{E}\big(V_{k+1}^m( X_{t_{k+1}}, Q) \rvert X_{t_k} \big) \big\|_2^2.
    \end{align*}
By the induction assumption, the first term in the right hand side converges to $0$ as $m \to + \infty$. Let us consider the second term. Since $\theta_{k+1, m}(Q)$ solves \eqref{optim_pb_nn_approx}, we have
\begin{align*}
\underset{Q \in \mathcal{T}_{k+1}}{\sup} \big\|\Phi_{m}\big(X_{t_k}; \theta_{k+1, m}(Q)\big) &- \mathbb{E}\big(V_{k+1}^m( X_{t_{k+1}}, Q) \rvert X_{t_k} \big) \big\|_2^2 \\
&\le \underset{Q \in \mathcal{T}_{k+1}}{\sup} \big\|\Phi_{m}\big(X_{t_k}; \tilde{\theta}_{k+1, m}(Q)\big) - \mathbb{E}\big(V_{k+1}^m( X_{t_{k+1}}, Q) \rvert X_{t_k} \big) \big\|_2^2,
\end{align*}
where $\tilde{\theta}_{k+1, m}(Q)$ solves the \q{theoretical} optimization problem,
$$\underset{\theta \in \Theta_m}{\inf} \hspace{0.1cm} \Big\|V_{k+1}\big(X_{t_{k + 1}}, Q\big) - \Phi_{m}\big(X_{t_k}; \theta  \big) \Big\|_2$$
with $\Theta_m$ defined in \eqref{theta_nn}. Then it follows from Minskowki inequality that
\small
\begin{align*}
\underset{Q \in \mathcal{T}_{k+1}}{\sup} \big\|\Phi_{m}\big(X_{t_k}; \tilde{\theta}_{k+1, m}(Q)\big) - \mathbb{E}\big(V_{k+1}^m( X_{t_{k+1}}, Q) \rvert X_{t_k} \big) \big\|_2^2 &\le 
\underset{Q \in \mathcal{T}_{k+1}}{\sup} \big\|\mathbb{E}\big(V_{k+1}( X_{t_{k+1}}, Q) \rvert X_{t_k} \big) - \mathbb{E}\big(V_{k+1}^m( X_{t_{k+1}}, Q) \rvert X_{t_k} \big) \big\|_2^2 \\
&\quad +\underset{Q \in \mathcal{T}_{k+1}}{\sup} \big\|\Phi_{m}\big(X_{t_k}; \tilde{\theta}_{k+1, m}(Q)\big) - \mathbb{E}\big(V_{k+1}( X_{t_{k+1}}, Q) \rvert X_{t_k} \big) \big\|_2^2.
\end{align*}

\normalsize

\noindent
Once again, by the induction assumption, the first term in the right hand side converges to $0$ as $m \to +\infty$. Moreover, thanks to the universal approximation theorem, for all $Q \in \mathcal{T}_{k+1}$
\begin{equation}
\label{eq_help_2}
\big\|\Phi_{m}\big(X_{t_k}; \tilde{\theta}_{k+1, m}(Q)\big) - \mathbb{E}\big(V_{k+1}( X_{t_{k+1}}, Q) \rvert X_{t_k} \big) \big\|_2^2 \xrightarrow[m \to + \infty]{} 0.
\end{equation}
Besides notice that,
\begin{align}
\label{eq_help}
 \big\|\Phi_{m}\big(X_{t_k}; \tilde{\theta}_{k+1, m}(Q)\big) - \mathbb{E}\big(V_{k+1}( X_{t_{k+1}}, Q) \rvert X_{t_k} \big) \big\|_2^2 &= 
 \underset{\Phi \in \mathcal{NN}_m}{\inf} \big\|\Phi\big(X_{t_k}\big) - \mathbb{E}\big(V_{k+1}(X_{t_{k+1}}, Q) \rvert X_{t_k} \big) \big\|_2^2,
\end{align}
where $\mathcal{NN}_m$ is defined in \eqref{nn_param_m}. But since the sequence $\big(\Theta_m \big)_m$ is non-decreasing (in the sense that $\Theta_m \subseteq \Theta_{m+1}$), then $\big(\mathcal{NN}_m\big)_m$ is too. So that by the previous equality \eqref{eq_help},
$$\Big( \big\|\Phi_{m}\big(X_{t_k}; \tilde{\theta}_{k+1, m}(Q)\big) - \mathbb{E}\big(V_{k+1}( X_{t_{k+1}}, Q) \rvert X_{t_k} \big) \big\|_2^2 \Big)_{m \ge 2}$$
is a non-increasing sequence. Thus keeping in mind equation \eqref{eq_help}, if the function,
\begin{align*}
  H_k \colon \big(\mathcal{NN}_m, \hspace{0.1cm} |\cdot|_{\sup}\big) \times \big(\mathcal{T}_{k+1}, \hspace{0.1cm} |\cdot|\big) &\to \mathbb{R}\\
  (\Phi, Q) &\longmapsto \|L_k(\Phi, Q)\|_2^2 := \big\|\Phi\big(X_{t_k}\big) - \mathbb{E}\big(V_{k+1}(X_{t_{k+1}}, Q) \rvert X_{t_k} \big) \big\|_2^2
\end{align*}
is continuous, then thanks to Theorem \ref{inf_comp} (noticing that for all $m \ge 2$, $\mathcal{NN}_m$ is a compact set), the function
$$Q \mapsto \big\|\Phi_{m}\big(X_{t_k}; \tilde{\theta}_{k+1, m}(Q)\big) - \mathbb{E}\big(V_{k+1}( X_{t_{k+1}}, Q) \rvert X_{t_k} \big) \big\|_2^2$$
will be continuous on the compact set $\mathcal{T}_{k+1}$. Thus one may use Dini lemma and conclude that the pointwise convergence in \eqref{eq_help_2} is in fact uniform. 
Which will completes the proof. 

\vspace{0.2cm}
Note that we have already shown that $Q \mapsto \mathbb{E}\big(V_{k+1}(X_{t_{k+1}}, Q) \rvert X_{t_k} \big)$ is almost surely continuous under assumption $\bm{\mathcal{H}_{3, 2q}}$. Moreover using the classic inequality: $(a+b)^2 \le 2(a^2 + b^2)$ and then conditional Jensen inequality
\begin{align*}
\big|L_k(\Phi, Q)|^2 &\le 2 \cdot \big|\Phi\big(X_{t_k}\big)\big|^2 + 2 \cdot \mathbb{E}\big(V_{k+1}(X_{t_{k+1}}, Q)^2 \rvert X_{t_k} \big)\\
&\le 2 \cdot \big|\Phi\big(X_{t_k}\big)\big|^2 + 2 \cdot \mathbb{E}\big(G_{k+1}^2 \rvert X_{t_k} \big) \in \mathbb{L}^1_{\mathbb{R}}\big(\mathbb{P}\big),
\end{align*}

\noindent
where the existence of $G_{k+1} \in \mathbb{L}^2_{\mathbb{R}}(\mathbb{P})$ (independent of $Q$) follows from Remark \ref{cont_high_order} and is implied by assumption
 $\bm{\mathcal{H}_{3, 2q}}$. Note that the integrability of $\big|\Phi\big(X_{t_k}\big)\big|^2$ follows from assumptions $\bm{\mathcal{H}_1^{\mathcal{NN}}}$ and 
 $\bm{\mathcal{H}_{3, 2q}}$. This implies that $\|L_k(\Phi, \cdot)\|_2^2$ is continuous.

\vspace{0.2cm}

Besides, for some sequence $(\Phi_n)_n$ of $\mathcal{NN}_m$ such that $\Phi_n \xrightarrow[n \to + \infty]{|\cdot|_{\sup}} \Phi$, it follows from the Lebesgue's dominated 
convergence theorem (enabled by  assumptions $\bm{\mathcal{H}_1^{\mathcal{NN}}}$ and $\bm{\mathcal{H}_{3, 2q}}$) that 
$\|L_k(\Phi_n, Q)\|_2^2 \xrightarrow[n \to + \infty]{} \|L_k(\Phi, Q)\|_2^2$. Which shows that $\|L_k(\cdot, Q)\|_2^2$ is continuous. 
Therefore the function $H_k$ is continuous. And as already mentioned this completes the proof.
\end{proof}

\begin{remark}[Assumptions $\bm{\mathcal{H}_1^{\mathcal{NN}}}$ and $\bm{\mathcal{H}_2^{\mathcal{NN}}}$]
In the previous proposition, we made the assumption that the neural networks are continuous and with polynomial growth. This assumption is clearly satisfied when using classic activation functions such as the ReLU function $x \in \mathbb{R} \mapsto \max(x,0)$ and Sigmoïd function $x \in \mathbb{R} \mapsto 1 / (1 + e^{-x})$.
\end{remark}

\subsection{Convergence of Monte Carlo approximation}
From now on, we assume a fixed positive integer $m$ and our focus is on the convergence of the value function that arises from the second approximation 
\eqref{estim_second_approx} or \eqref{dp_approx_nn_mc}. Unlike the preceding section and for technical convenience, we restrict our analysis of the neural network 
approximation to the \textit{bang-bang} setting. However, the linear regression will still be examined in a general context.

Herefater in this section and in the rest of the paper, we slightly modify the notation of norms and expectations to make clear their dependence on the Monte Carlo sample size $N$. 
Consider an i.i.d. sample $\big(X^{(i)}\big)_{1 \le i \le N}$ of size $N$ of a random vector $X$. For any $N \ge 1$, we denote by $\|\cdot\|_{r, N}$ the $\mathbb{L}^r\big(\mathbb{P}_N\big)$-norm where $\mathbb{P}_N$ is empirical probility measure associated with the Monte
Carlo sample. Likewise, we denote by $\mathbb{E}_N$ the expectation with respect to the empirical measure $\mathbb{P}_N$. In the same spirit, we denote by
$\mathbb{P}_{\otimes N}$ the product measure of $(X^{(i)}, X)$, and by $\mathbb{E}_{\otimes N}$ the expectation with respect to the latter product measure.

\subsubsection{Linear regression regression}
We establish a convergence result under the following Hilbert assumption.

\vspace{0.4cm}
\noindent
$\bm{\mathcal{H}_5^{LS}}$: For all $k=0,\ldots,n-1$ the sequence $\left(e_i\left( X_{t_k} \right)\right)_{i \ge 1}$ is a Hilbert basis of $\mathbb{L}^2\big(\sigma(X_{t_k})  \big)$.

\vspace{0.3cm}

\noindent
It is worth noting that this assumption is a special case of assumptions $\bm{\mathcal{H}_1^{LS}}$ and $\bm{\mathcal{H}_2^{LS}}$ with an orthonormality assumption on $e^m\big(X_{t_k}\big)$. Furthermore, in the field of mathematical finance, the underlying asset's diffusion is often assumed to have a Gaussian structure. However, it is well known that the normalized Hermite polynomials $\big\{\frac{H_k(x)}{\sqrt{k!}} , k \ge 0 \big\}$ serve as a Hilbert basis for $\mathbb{L}^2(\mathbb{R}, \mu)$, the space of square-integrable functions with respect to the Gaussian measure $\mu$. The Hermite polynomials $\big\{ H_k(x), k \ge 0\big\}$ are defined as follows:
$$H_k(x) = (-1)^{k} e^{x^2} \frac{d^k}{dx^k} \big[ e^{-x^2} \big],$$

\noindent
or recursively by
$$H_{k+1}(x) = 2x \cdot H_k(x) - 2k \cdot H_{k-1}(x) \hspace{0.4cm} \text{with} \hspace{0.4cm} H_0(x) = 1, \hspace{0.2cm} H_1(x) = 2x.$$

\noindent
For a multidimensional setting, Hermite polynomials are obtained as the product of one-dimensional Hermite polynomials. Finally, note that assumptions $\bm{\mathcal{H}_5^{LS}}$ entail that $A_m^k = A_{m, N}^k = I_m$.

\vspace{0.3cm}
The main result of this section aims at proving that the second approximation $V_k^{m, N}$ of the swing value function converges towards the first approximation $V_k^m$ as the Monte Carlo sample size $N$ increases to $+\infty$ and with a rate of convergence of order $\mathcal{O}\big(\frac{1}{\sqrt{N}}\big)$. To achieve this we rely on the following lemma which concern general Monte Carlo rate of convergence.

\begin{lemma}[Monte Carlo $\mathbb{L}^r\big(\mathbb{P}\big)$-rate of convergence]
\label{lem_zyg}
Consider $X_1, \ldots, X_N$ independent and identically distributed random variables with order $p$ ($p \ge 2$) finite moment (with $\mu = \mathbb{E}(X_1)$). 
Then, there exists a positive constant $B_p$ (only depending on the order $p$) such that
$$\Big\|\frac{1}{N} \sum_{i = 1}^{N} X_i - \mu \Big\|_{p, N}\le B_p \frac{2^{\frac{p-1}{p}} \Big(\mathbb{E}(|X|^p) + |\mu|^p \Big)^{\frac{1}{p}}}{\sqrt{N}}.$$
\end{lemma}

\begin{proof}
It follows from Marcinkiewicz–Zygmund inequality that there exists a positive constant $A_p$ (only depends on $p$) such that
\begin{align*}
\Big\|\frac{1}{N} \sum_{i = 1}^{N} X_i - \mu \Big\|_{p, N}^p = \mathbb{E}_N\left(\Big(\sum_{i = 1}^{N} \frac{X_i - \mu}{N} \Big)^{p}\right)
\le A_p \cdot \mathbb{E}_N\left(\Big(\frac{1}{N^2} \sum_{i = 1}^{N} (X_i - \mu)^2 \Big)^{p/2}\right)\\
= \frac{A_p}{N^{\frac{p}{2}}} \cdot  \mathbb{E}_N\left(\Big(\frac{1}{N} \sum_{i = 1}^{N} (X_i - \mu)^2 \Big)^{p/2}\right).
\end{align*}
Using the convexity of the function $x \in \mathbb{R}_{+} \mapsto x^{p/2}$ yields,
$$\Big(\frac{1}{N} \sum_{i = 1}^{N} (X_i - \mu)^2 \Big)^{p/2} \le \frac{1}{N} \sum_{i = 1}^{N} (X_i - \mu)^p.$$
Thus taking the expectation and using the inequality, $(a+b)^p \le 2^{p-1}(a^p + b^p)$ yields,
$$\Big\|\frac{1}{N} \sum_{i = 1}^{N} X_i - \mu \Big\|_{p,N}^p \le \frac{A_p}{N^{\frac{p}{2}}} \cdot \mathbb{E}\Big((X - \mu)^p \Big) \le A_p \cdot \frac{2^{p-1} \Big(\mathbb{E}(|X|^p) + |\mu|^p\Big)}{N^{\frac{p}{2}}}.$$
This completes the proof.
\end{proof}

In the following proposition, we show that using Hilbert basis as a regression basis allows to achieve a convergence with a rate of order $\mathcal{O}(\frac{1}{\sqrt{N}})$.

\begin{Proposition}
\label{cvg_hilbert_bis}
Under assumptions $\bm{\mathcal{H}_{3, \infty}}$, $\bm{\mathcal{H}_{4, \infty}^{LS}}$ and $\bm{\mathcal{H}_5^{LS}}$, for all $k \in \{0,\ldots,n-1\}$ and for any $s > 1$, we have
$$\underset{Q \hspace{0.1cm} \in \hspace{0.1cm} \mathcal{T}_{k}}{\sup} \hspace{0.2cm} \Big\| V^{m, N}_k\left(X_{t_k}, Q \right) - V^{m}_k\left(X_{t_k}, Q \right) \Big\|_{s,\otimes N} = \mathcal{O}\left(\frac{1}{\sqrt{N}}\right) \hspace{0.5cm} \text{as} \hspace{0.2cm} N \to +\infty.$$
\end{Proposition}

\begin{proof}
We prove this proposition using a backward induction on $k$. Since $V^{m, N}_{n-1}\left( X_{t_{n-1}}, \cdot \right) = V^{m}_{n-1}\left(X_{t_{n-1}}, \cdot \right)$ on $\mathcal{T}_{n-1}$, then the proposition holds for $k = n-1$. Assume now that the proposition holds for $k+1$. Using the inequality, $|\underset{i \in I}{\sup} \hspace{0.1cm} a_i - \underset{i \in I}{\sup} \hspace{0.1cm} b_i| \hspace{0.1cm} \le \hspace{0.1cm} \underset{i \in I}{\sup} \hspace{0.1cm} |a_i-b_i|$ and then Cauchy-Schwartz' one, we get,
\begin{align*}
\big|V^{m, N}_k\left(X_{t_{k}}, Q \right) - V^{m}_k(X_{t_{k}}, Q) \big| &\le  \esssup_{q \in \mathbb{A}_k(Q)} \big|\langle \theta_{k+1,m,N}(Q+q) - \theta_{k+1,m}(Q+q), e^m(X_{t_k})\rangle  \big|\\
&\le \big|e^m(X_{t_k}) \big| \cdot \esssup_{q \in \mathbb{A}_k(Q)} \hspace{0.1cm} \big|\theta_{k+1,m,N}(Q+q) - \theta_{k+1,m}(Q+q) \big|\\
&\le \big|e^m(X_{t_k}) \big| \cdot \esssup_{q \hspace{0.1cm} \in \hspace{0.1cm} \mathcal{U}_{k}(Q)} \hspace{0.1cm} \big|\theta_{k+1,m,N}(Q + q) - \theta_{k+1,m}(Q + q) \big|,
\end{align*}
where $\mathcal{U}_{k}(Q)$ is the set of all $\mathcal{F}_{{k+1}}^X$-measurable random variables lying within $\mathcal{I}_{k+1}\big(Q\big)$ (see \eqref{intervall_adm}). 
The last inequality is due to the fact that $\mathcal{F}_{k}^X \subset \mathcal{F}_{{k+1}}^X$. Then for some constants $b, c > 1$ such that $\frac{1}{b} + \frac{1}{c} = 1$, it follows from Hölder inequality that,
\begin{align}
\label{interm_res_cvg_mc}
\Big\|V^{m, N}_k(X_{t_{k}}, Q) - V^{m}_k(X_{t_{k}}, Q) \Big\|_{s,\otimes N} &\le \Big\| |e^m(X_{t_k}) |  \Big\|_{sb} \cdot 
\Big\| \esssup_{q \hspace{0.1cm} \in \hspace{0.1cm} \mathcal{U}_{k}(Q)} \hspace{0.1cm} \big|\theta_{k+1,m,N}(Q + q) - \theta_{k+1,m}(Q + q) \big|  \Big\|_{sc, \otimes N}.
\end{align}
To interchange the expectation and the essential supremum, we rely on the bifurcation property. Let $q_1, q_2 \in \mathcal{U}_{k}(Q)$ and denote by
$$q^{*} = q_1 \cdot \mathrm{1}_{\{B_k(Q, q_1) \ge B_k(Q, q_2)\}} + q_2 \cdot  \mathrm{1}_{\{B_k(Q, q_1) < B_k(Q, q_2)\}}$$
where $B_k(Q, q_i) = \big|\theta_{k+1,m,N}(Q+ q_i) - \theta_{k+1,m}(Q+ q_i)\big|^{sc}$ for $i \in \{1,2\}$. One can easily check that for all $i \in \{1,2\}$, $B_k(Q, q_i)$ is $\mathcal{F}_{t_{k+1}}^X$-measurable so that $q^{*} \in \mathcal{U}_{k}(Q)$. We also have $ B_k(Q, q^{*}) = \max\left(B_k(Q, q_1), B_k(Q, q_2) \right)$. Thus one may use the bifurcation property in \eqref{interm_res_cvg_mc}, we get,
\begin{align}
\label{my_eq_help2}
\Big\|V^{m, N}_k(X_{t_{k}}, Q) - V^{m}_k(X_{t_{k}}, Q) \Big|\Big|_{s,\otimes N} &\le \Big\| |e^m(X_{t_k})|  \Big|\Big|_{sb} \cdot \underset{q \in \mathcal{U}_{k}(Q)}{\sup} 
\hspace{0.1cm} \Big\|\theta_{k+1,m,N}(Q + q) - \theta_{k+1,m}(Q + q) |  \Big\|_{sc, \otimes N} \nonumber\\
&\le \Big\| |e^m(X_{t_k}) |  \Big\|_{sb} \cdot  \underset{Q \in \mathcal{T}_{k+1}}{\sup} \hspace{0.1cm} \Big|\Big||\theta_{k+1,m,N}(Q) - \theta_{k+1,m}(Q) |  \Big\|_{sc, \otimes N}.
\end{align}
But for any $Q \in \mathcal{T}_{k+1}$, it follows from Minkowski's inequality that,
\begin{align*}
\Big\||\theta_{k+1,m,N}(Q) - \theta_{k+1,m}(Q) |  \Big\|_{sc, \otimes  N} &= \Bigg\|  \Big| \frac{1}{N}\sum_{p = 1}^{N} e^m(X_{t_k}^{[p]}) \cdot V^{m, N}_{k+1}(X_{t_{k+1}}^{[p]}, Q) -  
\mathbb{E}\big(e^m(X_{t_k})V^{m}_{k+1}(X_{t_{k+1}}, Q)\big) \Big|  \Bigg\|_{sc, \otimes N}\\
&\le \Bigg\|\Big|\frac{1}{N}\sum_{p = 1}^{N} e^m(X_{t_k}^{[p]}) \cdot\left(V^{m, N}_{k+1}(X_{t_{k+1}}^{[p]}, Q) - V^{m}_{k+1}(X_{t_{k+1}}^{[p]}, Q)\right) \Big|  \Bigg\|_{sc, \otimes  N} \\
&\quad + \Bigg\|\Big|\frac{1}{N}\sum_{p = 1}^{N} e^m(X_{t_k}^{[p]})V^{m}_{k+1}( X_{t_{k+1}}^{[p]}, Q) - \mathbb{E}\left(e^m(X_{t_k})V^{m}_{k+1}(X_{t_{k+1}}, Q)\right)  \Big|\Bigg\|_{sc, \otimes N}\\
&\le \Bigg\|| e^m(X_{t_k})| \cdot |V^{m, N}_{k+1}(X_{t_{k+1}}, Q) - V^{m}_{k+1}(X_{t_{k+1}}, Q) |  \Bigg\|_{sc, \otimes N} \\
&\quad + \Bigg\|\Big|\frac{1}{N}\sum_{p = 1}^{N} e^m(X_{t_k}^{[p]})V^{m}_{k+1}( X_{t_{k+1}}^{[p]}, Q) - \mathbb{E}\left(e^m(X_{t_k})V^{m}_{k+1}(X_{t_{k+1}}, Q)\right)\Big|\Bigg\|_{sc, \otimes N},
\end{align*}
where the last inequality comes from the fact that, for all $p \ge 1$, $\big(X_{t_k}^{[p]},X_{t_{k+1}}^{[p]}  \big)$ has the same distribution with $\big(X_{t_k},X_{t_{k+1}} \big)$. Therefore, for some constants $u, v > 1$ such that $\frac{1}{u} + \frac{1}{v} = 1$, it follows from Hölder inequality,
\begin{align*}
\Big\| |\theta_{k+1,m,N}(Q) - \theta_{k+1,m}(Q) |  \Big|\Big|_{sc, \otimes N} &\le \Big\||e^m(X_{t_k})| \Big\|_{scu} \cdot 
\Big\|V^{m, N}_{k+1}(X_{t_{k+1}}, Q) - V^{m}_{k+1}(X_{t_{k+1}}, Q)  \Big\|_{scv, \otimes N} \\
&\quad + \Bigg\|\Big|\frac{1}{N}\sum_{p = 1}^{N} e^m(X_{t_k}^{[p]})V^{m}_{k+1}( X_{t_{k+1}}^{[p]}, Q) - 
\mathbb{E}\left(e^m(X_{t_k})V^{m}_{k+1}(X_{t_{k+1}}, Q)\right) \Big|  \Bigg\|_{sc, \otimes N}.
\end{align*}
Taking the supremum in the previous inequality and plugging it into equation \eqref{my_eq_help2} yields,
\begin{align*}
& \underset{Q \in \mathcal{T}_{k}}{\sup} \hspace{0.1cm} \Big\|V^{m, N}_k(X_{t_{k}}, Q ) - V^{m}_k(X_{t_{k}}, Q) \Big\|_{s,\otimes N}\\
&\le  \Big\| |e^m(X_{t_k}) |  \Big\|_{sb} \cdot \Big\|| e^m(X_{t_k})| \Big\|_{scu} \cdot \underset{Q \in \mathcal{T}_{k+1}}{\sup} \hspace{0.1cm}  
\Big\|V^{m, N}_{k+1}(X_{t_{k+1}}, Q) - V^{m}_{k+1}(X_{t_{k+1}}, Q) \Big\|_{scv, \otimes N}\\
&\quad + \Big\| |e^m(X_{t_k}) |  \Big\|_{sb} \cdot \underset{Q \in \mathcal{T}_{k+1}}{\sup} \hspace{0.1cm} 
\Bigg\|\Big|\frac{1}{N}\sum_{p = 1}^{N} e^m(X_{t_k}^{[p]})V^{m}_{k+1}(X_{t_{k+1}}^{[p]}, Q) - \mathbb{E}\left(e^m(X_{t_k})V^{m}_{k+1}(X_{t_{k+1}}, Q)\right)  \Big|  \Bigg\|_{sc, \otimes N}.
\end{align*}
Under assumption $\bm{\mathcal{H}_{4, r}^{LS}}$ and using induction assumption, the first term in the sum of the right hand side converges to 0 as $N \to +\infty$ with a rate of order $\mathcal{O}(\frac{1}{\sqrt{N}})$. Once again, by assumption $\bm{\mathcal{H}_{4, \infty}^{LS}}$, it remains to prove that it is also the case for the second term. But we have,
\begin{align*}
C_N(Q) &:= \Bigg\|\Big|\frac{1}{N}\sum_{p = 1}^{N} e^m(X_{t_k}^{[p]})V^{m}_{k+1}( X_{t_{k+1}}^{[p]}, Q) - 
\mathbb{E}\left(e^m(X_{t_k})V^{m}_{k+1}(X_{t_{k+1}}, Q)\right)  \Big|  \Bigg\|_{sc, \otimes N}\\
&=  \Bigg\|\sum_{j = 1}^{m} \left(\frac{1}{N}\sum_{p = 1}^{N} e_j(X_{t_k}^{[p]})V^{m}_{k+1}(X_{t_{k+1}}^{[p]}, Q) - 
\mathbb{E}\left(e_j(X_{t_k})V^{m}_{k+1}(X_{t_{k+1}}, Q)\right)  \right)^2  \Bigg\|_{\frac{sc}{2}, \otimes N}^{\frac{1}{2}}\\
&\le \sum_{j = 1}^{m} \Big\|\frac{1}{N}\sum_{p = 1}^{N} e_j(X_{t_k}^{[p]})V^{m}_{k+1}(X_{t_{k+1}}^{[p]}, Q) - \mathbb{E}\big(e_j(X_{t_k})V^{m}_{k+1}(X_{t_{k+1}}, Q)\big) \Big\|_{sc, \otimes N}\\
&\le \frac{A_{ac}}{\sqrt{N}} \cdot  \sum_{j = 1}^{m} \Big(\mathbb{E}\big(|e_j(X_{t_k})V^{m}_{k+1}(X_{t_{k+1}}, Q)|^{sc}\big) + 
\Big|\mathbb{E}\big(e_j(X_{t_k})V^{m}_{k+1}(X_{t_{k+1}}, Q)\big)\Big|^{sc} \Big),
\end{align*}
where the second-last inequality comes from Minkowski inequality and the inequality, $\sqrt{x+y} \le \sqrt{x} + \sqrt{y}$ for all $x, y \ge 0$. The last inequality is obtained using Lemma \ref{lem_zyg} (with a positive constant $A_{ac}$ only depends on the order $a$ and $c$). But using the continuity (which holds as noticed in Remark \ref{cont_high_order}) of both functions $Q \mapsto \mathbb{E}\Big(\big|e_j(X_{t_k})V^{m}_{k+1}(X_{t_{k+1}}, Q)\big|^{sc}\Big)$ and $Q \mapsto \big|\mathbb{E}(e_j(X_{t_k})V^{m}_{k+1}( X_{t_{k+1}}, Q))\big|^{sc}$ on the compact set $\mathcal{T}_{k+1}$ one may deduce that, as $N \to +\infty$,
$$\underset{Q \in \mathcal{T}_{k+1}}{\sup} \hspace{0.1cm} C_N(Q) = \mathcal{O}\left(\frac{1}{\sqrt{N}} \right).$$
This completes the proof.
\end{proof}

\begin{remark}[Almost surely convergence]
It is worth noting that it is difficult to obtain an almost surely convergence result without further assumptions (for example boundedness assumption) of the regression functions. The preceding proposition is widely based on Hölder inequality emphasizing on why we have chosen the $\mathbb{L}^s\big(\mathbb{P}\big)$-norm. However, in the neural network analysis that follows, we prove an almost surely convergence result.
\end{remark}

\subsubsection{Neural network approximation}
We consider the discrete setting with integer volume constraints with a state of attainable cumulative consumptions given by \eqref{range_cum_vol_discr}. Results in this section will be mainly based on Lemmas \ref{cvg_min_lem} and \ref{uslln} stated below. Let $(f_n)_n$ be a sequence of real functions defined on a compact set $K \subset \mathbb{R}^d$. Define,
$$v_n = \underset{x \in K}{\inf} \hspace{0.1cm} f_n(x)  \hspace{0.6cm} \text{and} \hspace{0.6cm} x_n \in \arginf_{x \in K} f_n(x).$$

\noindent
Then, we have the following two Lemmas.

\begin{lemma}[Convergence of minimizers]
\label{cvg_min_lem}
Assume that the sequence $(f_n)_n$ converges uniformly on $K$ to a continuous function $f$. Let $v^* = \underset{x \in K}{\inf} \hspace{0.1cm} f_n(x)$ and $\mathcal{S}^* = \arginf_{x \in K} f(x)$. Then $v_n \to v^*$ and the distance $d(x_n, \mathcal{S}^*)$ between the minimizer $x_n$ and the set $\mathcal{S}^*$ converges to $0$ as $n \to +\infty$.
\end{lemma}

\begin{lemma}[Uniform law of large numbers]
\label{uslln}
Let $(\xi_i)_{i \ge 1}$ be a sequence of i.i.d. $\mathbb{R}^m$-valued random vectors and $h : \mathbb{R}^d \times \mathbb{R}^m \to \mathbb{R}$ a measurable function. Assume that,

\begin{itemize}
\item a.s., $\theta \in \mathbb{R}^d \mapsto h(\theta, \xi_1)$ is continuous,
\item For all $C > 0$, $\mathbb{E}\Big( \underset{|\theta| \le C}{\sup} \hspace{0.1cm} \big| h(\theta, \xi_1) \big| \Big) <  +\infty$.
\end{itemize}

\noindent
Then, a.s. $\theta \in \mathbb{R}^d \mapsto \frac{1}{N} \sum_{i = 1}^{N} h(\theta, \xi_i)$ converges locally uniformly to the continuous function $\theta \in \mathbb{R}^d \mapsto \mathbb{E}\big(h(\theta, \xi_1) \big)$, i.e.
$$\lim\limits_{N \rightarrow +\infty} \hspace{0.1cm} \underset{[\theta| \le C}{\sup} \hspace{0.1cm} \Big| \frac{1}{N} \sum_{i = 1}^{N} h(\theta, \xi_i) - \mathbb{E}\big(h(\theta, \xi_1) \big) \Big| = 0 \hspace{0.2cm} \text{a.s.}$$
\end{lemma}

Combining the two preceding lemmas is the main tool to analyze the Monte Carlo convergence of the neural network approximation. The result is stated below and requires the following (additional) assumption.

\vspace{0.4cm}
\noindent
$\bm{\mathcal{H}_3^{\mathcal{NN}}}$: For any $m \ge 2$, $0 \le k \le n-1$, $Q \in \mathcal{T}_k$ and $\theta^1, \theta^2 \in \mathcal{S}_{k}^{m}(Q)$ (defined in \eqref{optim_pb_nn_approx}), $\Phi_m(\cdot; \theta^1) = \Phi_m(\cdot; \theta^2)$.

\vspace{0.2cm}
\noindent
This assumption just states that, almost surely, two minimizers bring the same value.

\vspace{0.3cm}
Before showing the main result of this section, it is worth noting this important remark.

\begin{remark}
\label{moment_ordre_vm_nn}
\begin{countlist}[label={(\Alph*)}]{otherlist}
  \item \label{born_inte_vm} Under assumptions $\bm{\mathcal{H}_1^{\mathcal{NN}}}$ and $\bm{\mathcal{H}_{3, q}}$ and using a straightforward backward induction in equation \eqref{dp_approx_nn}, it can be shown that there exists a random variable $G_k \in \mathbb{L}^q_{\mathbb{R}^d}\big(\mathbb{P}\big)$ (independent of $Q$) such that $\big|V_k^m(X_{t_k}, Q) \big| \le G_k$ for any $Q \in \mathcal{T}_k$; where $V_k^m$ is defined in \eqref{dp_approx_nn}.
  
  \item \label{born_vm_vnm} Under assumption $\bm{\mathcal{H}_1^{\mathcal{NN}}}$, there exists a positive constant $\kappa_m$ such that, for any $0 \le k \le n-1$ and any $Q \in \mathcal{T}_k$,
$$\max\Big(\big|V_k^m(X_{t_k}, Q)\big|, \big|V_k^{m, N}(X_{t_k}, Q)\big|  \Big) \le \overline{q}\cdot \big|S_{t_k} - K\big| + \kappa_m \cdot \big(1 + \big|X_{t_k}\big|^q \big).$$
If in addition, assumption $\bm{\mathcal{H}_{3, q}}$ holds true, then the right hand side of the last inequality is an integrable random variable.
\end{countlist}
\end{remark}

\vspace{0.2cm}
We now state our result of interest.

\begin{Proposition}
Let $m \ge 2$. Under assumptions $\bm{\mathcal{H}_1^{\mathcal{NN}}}$, $\bm{\mathcal{H}_2^{\mathcal{NN}}}$, $\bm{\mathcal{H}_3^{\mathcal{NN}}}$ and $\bm{\mathcal{H}_{3, 2q}}$, for any $0 \le k \le n-1$, we have,
$$\lim\limits_{N \rightarrow +\infty} \hspace{0.1cm} \underset{Q \hspace{0.1cm} \in \hspace{0.1cm} \mathcal{T}_{k}}{\sup} \hspace{0.2cm} \left| V^{m, N}_k\left(X_{t_k}, Q \right) - V^{m}_k\left(X_{t_k}, Q \right) \right|= 0 \hspace{0.6cm} \text{a.s.}$$
\end{Proposition}

\noindent
Note that in $\bm{\mathcal{H}_{3, 2q}}$, parameters $q$ are that involved in assumption $\bm{\mathcal{H}_1^{\mathcal{NN}}}$. Recall that, the set $\mathcal{T}_{k}$ is the one of the discrete setting as discussed in \eqref{range_cum_vol_discr}.

\begin{proof}
We proceed by a backward induction on $k$. The proposition clearly holds true for $k = n-1$ since, almost surely, $V_{n-1}^{m, N}(X_{t_{n-1}}, \cdot) = V_{n-1}^{m}(X_{t_{n-1}}, \cdot)$ on $\mathcal{T}_{n-1}$. Assume now the proposition holds true for $k+1$. Let $Q \in \mathcal{T}_k$. Using the inequality, $|\underset{i \in I}{\sup} \hspace{0.1cm} a_i - \underset{i \in I}{\sup} \hspace{0.1cm} b_i| \hspace{0.1cm} \le \hspace{0.1cm} \underset{i \in I}{\sup} \hspace{0.1cm} |a_i-b_i|$ and then triangle inequality, we get,
\begin{align}
\label{first_step_mc_nn}
\left| V^{m, N}_k\left(X_{t_k}, Q \right) - V^{m}_k\left(X_{t_k}, Q \right) \right| &\le \esssup_{q \in \mathbb{A}_k(Q)} \hspace{0.1cm} \Big|\Phi_m\big(X_{t_k}; \theta_{k, m, N}(Q+q) \big) - \Phi_m\big(X_{t_k}; \widetilde{\theta}_{k, m, N}(Q+q) \big) \Big| \nonumber \\
&\quad + \esssup_{q \in \mathbb{A}_k(Q)} \hspace{0.1cm} \Big|\Phi_m\big(X_{t_k}; \widetilde{\theta}_{k, m, N}(Q+q) \big) - \Phi_m\big(X_{t_k}; \theta_{k, m}(Q+q) \big) \Big|,
\end{align}
where $\widetilde{\theta}_{k, m, N}(Q)$ lies within the following set,
\begin{equation}
\label{th_set_min_proof_mc}
 \arginf_{\theta \in \Theta_m} \hspace{0.1cm} \frac{1}{N} \sum_{p = 1}^{N} \Big|V^{m}_{k+1}\big(X_{t_{k+1}}^{[p]}, Q \big)  - \Phi_m\big(X_{t_k}^{[p]}; \theta\big) \Big|^2.
\end{equation}
Then taking the supremum in \eqref{first_step_mc_nn} and using triangle inequality, we get,
\begin{align}
\label{sec_step_mc_nn}
\underset{Q \in \mathcal{T}_k}{\sup} \left| V^{m, N}_k\left(X_{t_k}, Q \right) - V^{m}_k\left(X_{t_k}, Q \right) \right| &\le 
\underset{Q \in \mathcal{T}_{k+1}}{\sup} \hspace{0.1cm} \Big|\Phi_m\big(X_{t_k}; \theta_{k, m, N}(Q) \big) - \Phi_m\big(X_{t_k}; \theta_{k, m}(Q) \big) \Big| \nonumber \\
&\quad + 2 \underset{Q \in \mathcal{T}_{k+1}}{\sup} \hspace{0.1cm} \Big|\Phi_m\big(X_{t_k}; \widetilde{\theta}_{k, m, N}(Q) \big) - \Phi_m\big(X_{t_k}; \theta_{k, m}(Q) \big) \Big|.
\end{align}
We will handle the right hand side of the last inequality term by term. Let us start with the second term. Note that owing to assumption $\bm{\mathcal{H}_2^{\mathcal{NN}}}$, the function
$$\theta \in \Theta_m \mapsto V^{m}_{k+1}\big(X_{t_{k+1}}, Q \big)  - \Phi_m\big(X_{t_k}; \theta\big)$$

\noindent
is almost surely continuous. Moreover, for any $C >0$, using the inequality $(a+b)^2 \le 2(a^2 + b^2)$ and assumption $\bm{\mathcal{H}_1^{\mathcal{NN}}}$, there exists a positive constant $\kappa_m$ such that for any $Q \in \mathcal{T}_{k+1}$,
\begin{align*}
\mathbb{E}\left(\underset{|\theta| \le C}{\sup} \hspace{0.1cm} \Big| V^{m}_{k+1}\big(X_{t_{k+1}}, Q \big)  - \Phi_m\big(X_{t_k}; \theta\big) \Big|^2 \right) &\le 2 \cdot \mathbb{E}\Big(\big| V^{m}_{k+1}\big(X_{t_{k+1}}, Q \big) \big|^2 \Big) + 2 \cdot \underset{|\theta| \le C}{\sup} \hspace{0.1cm}  \mathbb{E}\Big(\big|\Phi_m\big(X_{t_k}; \theta\big) \big|^2 \Big) \\
&\le 2 \cdot \mathbb{E}\Big(\big| V^{m}_{k+1}\big(X_{t_{k+1}}, Q \big) \big|^2 \Big) + 2\kappa_m \Big(1 + \mathbb{E}\big|X_{t_k}\big|^{2q}  \Big)
\end{align*}

\noindent
and the right hand side of the last inequality is finite under assumption $\bm{\mathcal{H}_{3, 2q}}$, keeping in mind point \ref{born_inte_vm} of Remark \ref{moment_ordre_vm_nn}. Thus thanks to Lemma \ref{uslln}, almost surely, we have the uniform convergence on $\Theta_m$, 
\begin{equation}
\label{first_unif_cvg}
\lim\limits_{N \rightarrow +\infty} \hspace{0.1cm} \underset{\theta \in \Theta_m}{\sup} \hspace{0.1cm} \left|  \frac{1}{N} 
\sum_{p = 1}^{N} \Big|V^{m}_{k+1}\big(X_{t_{k+1}}^{[p]}, Q \big)  - \Phi_m\big(X_{t_k}^{[p]}; \theta\big) \Big|^2 - \Big\| V^{m}_{k+1}\big(X_{t_{k+1}}, Q \big)  - 
\Phi_m\big(X_{t_k}; \theta\big) \Big\|_2^2  \right| = 0.
\end{equation}
Thus, for any $Q \in \mathcal{T}_{k+1}$, Lemma \ref{cvg_min_lem} implies that $\lim\limits_{N \rightarrow +\infty} \hspace{0.1cm} d\big(\widetilde{\theta}_{k, m, N}(Q), \mathcal{S}_k^m(Q)  \big) = 0$. We restrict ourselves to a subset with probability one of the original probability space on which this convergence holds and the random functions $\Phi_m\big(X_{t_k}; \cdot\big)$ are uniformly continuous (see assumption $\bm{\mathcal{H}_2^{\mathcal{NN}}}$). Then, there exists a sequence $\big(\alpha_{k, m, N}(Q)\big)_N$ lying within $\mathcal{S}_k^m(Q) $ such that,
$$\lim\limits_{N \rightarrow +\infty} \hspace{0.1cm} \Big|\widetilde{\theta}_{k, m, N}(Q) - \alpha_{k, m, N}(Q) \Big| = 0.$$
Thus, the uniform continuity of functions $\Phi_m\big(X_{t_k}; \cdot\big)$ combined with assumption $\bm{\mathcal{H}_3^{\mathcal{NN}}}$ yield,
$$\Big|\Phi_m\big(X_{t_k}; \widetilde{\theta}_{k, m, N}(Q) \big) - \Phi_m\big(X_{t_k}; \theta_{k, m}(Q) \big) \Big| = 
\Big|\Phi_m\big(X_{t_k}; \widetilde{\theta}_{k, m, N}(Q) \big) - \Phi_m\big(X_{t_k}; \alpha_{k, m, N}(Q) \big) \Big|  \xrightarrow[N \to + \infty]{} 0.$$

\noindent
Furthermore, since the set $\mathcal{T}_{k+1}$ has a finite cardinal (discrete setting) then, we have
\begin{equation}
\label{firs_lim_unif}
\lim\limits_{N \rightarrow +\infty} \hspace{0.1cm} \underset{Q \in \mathcal{T}_{k+1}}{\sup} \hspace{0.1cm} \Big|\Phi_m\big(X_{t_k}; \widetilde{\theta}_{k, m, N}(Q) \big) - \Phi_m\big(X_{t_k}; \theta_{k, m}(Q) \big) \Big| = 0.
\end{equation}

\noindent
It remains to handle the first term in the right hand side of inequality \eqref{sec_step_mc_nn}. Note that, if the following uniform convergence,
\begin{equation}
\label{unif_cvg_interest}
\lim\limits_{N \rightarrow +\infty} \hspace{0.1cm} \underset{\theta \in \Theta_m}{\sup} \hspace{0.1cm} \underbrace{\left|  \frac{1}{N} \sum_{p = 1}^{N} \Big|V^{m, N}_{k+1}\big(X_{t_{k+1}}^{[p]}, Q \big)  - \Phi_m\big(X_{t_k}^{[p]}; \theta\big) \Big|^2 -  \frac{1}{N} \sum_{p = 1}^{N} \Big|V^{m}_{k+1}\big(X_{t_{k+1}}^{[p]}, Q \big)  - \Phi_m\big(X_{t_k}^{[p]}; \theta\big) \Big|^2   \right|}_{:=\big|\Delta_{k, m, N}^Q(\theta)\big|}= 0
\end{equation}

\noindent
holds true, then the latter uniform convergence will entail the following one  owing to the uniform convergence \eqref{first_unif_cvg},
\begin{equation}
\label{last_unif_cvg}
\lim\limits_{N \rightarrow +\infty} \hspace{0.1cm} \underset{\theta \in \Theta_m}{\sup} \hspace{0.1cm} \left|  
    \frac{1}{N} \sum_{p = 1}^{N} \Big|V^{m, N}_{k+1}\big(X_{t_{k+1}}^{[p]}, Q \big)  - \Phi_m\big(X_{t_k}^{[p]}; \theta\big) \Big|^2 - \Big\| 
    V^{m}_{k+1}\big(X_{t_{k+1}}, Q \big)  - \Phi_m\big(X_{t_k}; \theta\big) \Big\|_2^2  \right| = 0
\end{equation}

\noindent
and the desired result follows. To achieve this, we start by proving the uniform convergence \eqref{unif_cvg_interest}. Then we show how its implication \eqref{last_unif_cvg} entails the desired result.

\vspace{0.3cm}
Using triangle inequality and the elementary identity, $a^2 - b^2 = (a-b)(a+b)$, we have,
\begin{align*}
\big|\Delta_{k, m, N}^Q(\theta)\big| &\le \frac{1}{N} \sum_{p = 1}^{N} \Big|V^{m, N}_{k+1}\big(X_{t_{k+1}}^{[p]}, Q \big) + V^{m}_{k+1}\big(X_{t_{k+1}}^{[p]}, Q \big) - 2 \cdot  \Phi_m\big(X_{t_k}^{[p]}; \theta\big) \Big| \cdot \Big| V^{m, N}_{k+1}\big(X_{t_{k+1}}^{[p]}, Q \big) - V^{m}_{k+1}\big(X_{t_{k+1}}^{[p]}, Q \big) \Big|\\
&\le \frac{2}{N} \sum_{p = 1}^{N} \Big(\overline{q} \big|S_{t_{k+1}}^{[p]} - K\big| + \kappa_m \big(1 + |X_{t_{k+1}}^{[p]}|^q\big) + \kappa_m \big(1 + |X_{t_{k}}^{[p]}|^q\big) \Big) \cdot \Big| V^{m, N}_{k+1}\big(X_{t_{k+1}}^{[p]}, Q \big) - V^{m}_{k+1}\big(X_{t_{k+1}}^{[p]}, Q \big) \Big|,
\end{align*}

\noindent
where in the last inequality we used assumption $\bm{\mathcal{H}_1^{\mathcal{NN}}}$ and the point \ref{born_vm_vnm} of Remark \ref{moment_ordre_vm_nn}. Let $\varepsilon > 0$. Then using the induction assumption and the law of large numbers, we get,
$$\limsup\limits_{N} \underset{\theta \in \Theta_m}{\sup} \hspace{0.1cm} \big|\Delta_{k, m, N}^Q(\theta)\big|  \le 2\varepsilon \cdot \mathbb{E}\Big(\overline{q} \big|S_{t_{k+1}} - K\big| + \kappa_m \big(1 + |X_{t_{k+1}}|^q\big) + \kappa_m \big(1 + |X_{t_{k}}|^q\big) \Big).$$

\noindent
Hence letting $\varepsilon \to 0$ entails the result \eqref{unif_cvg_interest}. Therefore, as already mentioned, the result \eqref{last_unif_cvg} also holds true. Thus, using Lemma \ref{cvg_min_lem}, we get that $\lim\limits_{N \rightarrow +\infty} \hspace{0.1cm} d\big(\theta_{k, m, N}(Q), \mathcal{S}_k^m(Q)  \big) = 0$. We restrict ourselves to a subset with probability one of the original probability space on which this convergence holds and the random functions $\Phi_m\big(X_{t_k}; \cdot\big)$ are uniformly continuous (see assumption $\bm{\mathcal{H}_2^{\mathcal{NN}}}$). Whence, for any $Q \in \mathcal{T}_{k+1}$, there exists a sequence $\big(\beta_{k, m, N}(Q)\big)_N$ lying within $\mathcal{S}_k^m(Q) $ such that,
$$\lim\limits_{N \rightarrow +\infty} \hspace{0.1cm} \Big|\theta_{k, m, N}(Q) - \beta_{k, m, N}(Q) \Big| = 0.$$

\noindent
Thus, the uniform continuity of functions $\Phi_m\big(X_{t_k}; \cdot\big)$ combined with assumption $\bm{\mathcal{H}_3^{\mathcal{NN}}}$ yield,
$$\Big|\Phi_m\big(X_{t_k}; \theta_{k, m, N}(Q) \big) - \Phi_m\big(X_{t_k}; \theta_{k, m}(Q) \big) \Big| = \Big|\Phi_m\big(X_{t_k}; \theta_{k, m, N}(Q) \big) - \Phi_m\big(X_{t_k}; \beta_{k, m, N}(Q) \big) \Big|  \xrightarrow[N \to + \infty]{} 0.$$

\noindent
Then, since the set $\mathcal{T}_{k+1}$ has a finite cardinal (discrete setting), we have
\begin{equation}
\label{snc_lim_unif}
\lim\limits_{N \rightarrow +\infty} \hspace{0.1cm} \underset{Q \in \mathcal{T}_{k+1}}{\sup} \hspace{0.1cm} \Big|\Phi_m\big(X_{t_k}; \theta_{k, m, N}(Q) \big) - \Phi_m\big(X_{t_k}; \theta_{k, m}(Q) \big) \Big| = 0.
\end{equation}

\noindent
Combining equations \eqref{firs_lim_unif} and \eqref{snc_lim_unif} in equation \eqref{sec_step_mc_nn} yield the desired result.
\end{proof}

\subsection{Deviation inequalities: linear regression setting}
To end this paper, we present some additional results related to linear regression approximation. These results focus on some deviation inequalities on the error between estimates \eqref{estim_second_approx}, \eqref{estim_orth_proj_dp} and the swing actual value function \eqref{eq_dp_swing}. We no longer consider the Hilbert assumption $\bm{\mathcal{H}_5^{LS}}$. Let us start with the first proposition of this section.

\begin{Proposition}
\label{prob_diff}
Let $\delta > 0$ and $k = 0, \ldots, n-2$. Under assumptions $\bm{\mathcal{H}_{3, \infty}}$ and $\bm{\mathcal{H}_{4, \infty}^{LS}}$, for all $s \ge 2$, there exists a positive constant $D_{s, k, m}$ such that, 
$$\mathbb{P}\left(\esssup_{Q \in \mathcal{Q}_k} \left|\frac{1}{N} \sum_{p = 1}^{N} e^m(X_{t_k}^{[p]}) V^m_{k+1}(X_{t_{k+1}}^{[p]}, Q) - \mathbb{E}\big(e^m(X_{t_k}) V^m_{k+1}( X_{t_{k+1}}, Q) \big)\right| \ge \delta \right) \le \frac{D_{s, k, m}}{\delta^s N^\frac{s}{2}}$$

\noindent
where $\mathcal{Q}_k$ is the set of all $\mathcal{F}_{t_k}^X$-measurable random variables lying within $\mathcal{T}_{k+1}$.

\end{Proposition}

\begin{proof}
Note that $\mathcal{Q}_k \subset \mathcal{Q}_k'$; with the latter set being the set of all $\mathcal{F}_{t_{k+1}}^X$-measurable random variables lying within $\mathcal{T}_{k+1}$. Then we have,
\begin{align*}
& \mathbb{P}\left(\esssup_{Q \in \mathcal{Q}_k} \Bigg|\frac{1}{N} \sum_{p = 1}^{N} e^m(X_{t_k}^{[p]}) V^m_{k+1}( X_{t_{k+1}}^{[p]}, Q) - \mathbb{E}\big(e^m(X_{t_k}) V^m_{k+1}(X_{t_{k+1}}, Q) \big)\right| \ge \delta \Bigg)\\
&\le \mathbb{P}\left(\esssup_{Q \in \mathcal{Q}_k'} \left|\frac{1}{N} \sum_{p = 1}^{N} e^m(X_{t_k}^{[p]}) V^m_{k+1}(X_{t_{k+1}}^{[p]}, Q) - \mathbb{E}\big(e^m(X_{t_k}) V^m_{k+1}(X_{t_{k+1}}, Q) \big)\right| \ge \delta \right)\\
&\le A_s \frac{\underset{Q \in \mathcal{Q}_k'}{\sup} \hspace{0.1cm} \Big\{ \mathbb{E}\Big(\big| e^m(X_{t_k}) V^m_{k+1}(X_{t_{k+1}}, Q) \big|^s \Big) + \mathbb{E}\Big(\big| e^m(X_{t_k}) V^m_{k+1}(X_{t_{k+1}}, Q) \big| \Big)^s \Big\}}{N^{s/2} \cdot \delta^s}\\
&\le 2A_s \frac{\underset{Q \in \mathcal{Q}_k'}{\sup} \hspace{0.1cm}  \mathbb{E}\Big(\big| e^m(X_{t_k}) V^m_{k+1}(X_{t_{k+1}}, Q) \big|^s \Big)}{N^{s/2} \cdot \delta^s},
\end{align*}

\noindent
where in the second-last inequality, we successively used Markov inequality, bifurcation property and Lemma \ref{lem_zyg} (enabled by assumptions $\bm{\mathcal{H}_{3, \infty}}$ and $\bm{\mathcal{H}_{4, \infty}^{LS}}$) with $A_s = B_s^s \cdot 2^{s-1}$ and $B_s$ being a positive constant which only depends on $a$. To obtain the last inequality, we used Jensen inequality. Besides, following the definition of $\mathcal{Q}_k'$ we have,
$$\underset{Q \in \mathcal{Q}_k'}{\sup} \hspace{0.1cm}  \mathbb{E}\Big(\big| e^m(X_{t_k}) V^m_{k+1}(X_{t_{k+1}}, Q) \big|^s \Big) \le \underset{Q \in \mathcal{T}_{k+1}}{\sup} \hspace{0.1cm}  \mathbb{E}\Big(\big| e^m(X_{t_k}) V^m_{k+1}(X_{t_{k+1}}, Q) \big|^s \Big).$$

\noindent
Then owing to Remark \ref{cont_high_order}, the right hand side of the last inequality is a supremum of a continuous function over a compact set; thus finite. Hence it suffices to set,
$$D_{s, k, m} := 2A_s \cdot \underset{Q \in \mathcal{T}_{k+1}}{\sup} \hspace{0.1cm}  \mathbb{E}\Big(\big| e^m(X_{t_k}) V^m_{k+1}(X_{t_{k+1}}, Q) \big|^s \Big) < + \infty.$$

\noindent
Which completes the proof.
\end{proof}

In the following proposition, we state a deviation inequality connecting the estimates of the orthogonal projection coordinates involved in the linear regression.

\begin{Proposition}
\label{tight_coeff_reg}
Consider assumptions $\bm{\mathcal{H}_{3, \infty}}$ and $\bm{\mathcal{H}_{4, \infty}^{LS}}$. For all $k=0, \ldots, n-2$, $\delta > 0$ and $s \ge 2$ there exists a positive constant $C_{s, k, m}$ such that,
$$\mathbb{P}\left(\esssup_{Q \in \mathcal{Q}_k} \hspace{0.1cm} \Big|\theta_{k, m, N}(Q) - \theta_{k, m}(Q)  \Big| \ge \delta \right) \le \frac{C_{s, k, m}}{b(s, \delta) \cdot N^\frac{s}{2}},$$

\noindent
where $b(s, \delta) = \delta ^s$ if $\delta \in (0,1]$ else $b(s, \delta) = \delta ^{s/2}$.
\end{Proposition}

\begin{proof}
We proceed by a backward induction on $k$. Recall that, for any $Q \in \mathcal{T}_{n-1}$, $V_{n-1}^{m, N}(\cdot, Q) = V_{n-1}^m(\cdot, Q)$. Thus, it follows from triangle inequality,
\begin{align*}
\Big|\theta_{n-2, m, N}(Q) - \theta_{n-2, m}(Q) \Big| &= \Big| \big(A_{m, N}^{n-2}\big)^{-1} \frac{1}{N} \sum_{p = 1}^{N} e^m(X_{t_{n-2}}^{[p]}) V^m_{n-1}(X_{t_{n-1}}^{[p]}, Q) -  \big(A_{m}^{n-2}\big)^{-1} \mathbb{E}\left(e^m(X_{t_{n-2}}) V^m_{n-1}(X_{t_{n-1}}, Q) \right)\Big|\\
&\le \Big|\big(A_{m, N}^{n-2}\big)^{-1} \Big(\frac{1}{N} \sum_{p = 1}^{N} e^m(X_{t_{n-2}}^{[p]}) V^m_{n-1}(X_{t_{n-1}}^{[p]}, Q) - \mathbb{E}\big(e^m(X_{t_{n-2}}) V^m_{n-1}( X_{t_{n-1}}, Q) \big)  \Big) \Big| \\
&+ \left|\left(\big(A_{m, N}^{n-2}\big)^{-1} - \big(A_{m}^{n-2}\big)^{-1} \right) \cdot \mathbb{E}\big(e^m(X_{t_{n-2}}) V^m_{n-1}(X_{t_{n-1}}, Q) \big) \right|\\
&= \Big|\big(A_{m, N}^{n-2}\big)^{-1} \Big(\frac{1}{N} \sum_{p = 1}^{N} e^m(X_{t_{n-2}}^{[p]}) V^m_{n-1}(X_{t_{n-1}}^{[p]}, Q) - \mathbb{E}(e^m(X_{t_{n-2}}) V^m_{n-1}(X_{t_{n-1}}, Q))  \Big) \Big| \\
&+ \left|\left(\big(A_{m}^{n-2}\big)^{-1}\big(A_{m}^{n-2} - A_{m, N}^{n-2}\big)\big(A_{m, N}^{n-2}\big)^{-1} \right) \mathbb{E}\left(e^m(X_{t_{n-2}}) V^m_{n-1}(X_{t_{n-1}}, Q) \right) \right|,
\end{align*}

\noindent
where in the last equality we used the matrix identity $A^{-1} - B^{-1} = B^{-1} (B-A) A^{-1}$ for all non-singular matrices $A, B$. Hence taking the essential supremum and keeping in mind that the matrix norm $|\cdot|$ is submultiplicative yields,
\begin{align*}
& \esssup_{Q \in \mathcal{Q}_{n-2}} \hspace{0.1cm} \Big|\theta_{n-2, m, N}(Q) - \theta_{n-2, m}(Q) \Big| \\
&\le \Big|\big(A_{m, N}^{n-2}\big)^{-1} \Big| \cdot \esssup_{Q \in \mathcal{Q}_{n-2}} \hspace{0.1cm} \Big|\frac{1}{N} \sum_{p = 1}^{N} e^m(X_{t_{n-2}}^{[p]}) V^m_{n-1}(X_{t_{n-1}}^{[p]}, Q) - \mathbb{E}\big(e^m(X_{t_{n-2}}) V^m_{n-1}(X_{t_{n-1}}, Q)\big) \Big|\\
&\quad + C_{n-2} \cdot \Big|\big(A_{m}^{n-2}\big)^{-1}\big(A_{m}^{n-2} - A_{m, N}^{n-2}\big)\big(A_{m, N}^{n-2}\big)^{-1}\Big|,
\end{align*}

\noindent
where $C_{n-2} :=  \underset{Q \in \mathcal{T}_{n-1}}{\sup} \hspace{0.1cm} \big|\mathbb{E}\left(e^m(X_{t_{n-2}}) V^m_{n-1}(X_{t_{n-1}}, Q) \right)\big| < +\infty$. For any $\varepsilon > 0$ and $k = 0, \ldots, n-2$, denote by $\Omega_k^{\varepsilon} := \big\{\big|A_{m, N}^{k} - A_{m}^{k} \big| \le \varepsilon \big\}$. Then one may choose $\varepsilon$ such that $\big|(A_{m, N}^{k})^{-1}\big| \le 2 \big|(A_{m}^{k})^{-1} \big|$ on $\Omega_{k}^{\varepsilon}$. Thus there exists positive constants $K_1, K_2$ such that on $\Omega_{n-2}^{\varepsilon}$,
\begin{align*}
& \esssup_{Q \in \mathcal{Q}_{n-2}} \hspace{0.1cm} \Big|\theta_{n-2, m, N}(Q) - \theta_{n-2, m}(Q) \Big|\\
&\le  K_1 \cdot \esssup_{Q \in \mathcal{Q}_{n-2}} \hspace{0.1cm}\Big| \frac{1}{N} \sum_{p = 1}^{N} e^m(X_{t_{n-2}}^{[p]}) V^m_{n-1}(X_{t_{n-1}}^{[p]}, Q) - \mathbb{E}\big(e^m(X_{t_{n-2}}) V^m_{n-1}( X_{t_{n-1}}, Q) \big) \Big| + K_2 \cdot \varepsilon.
\end{align*}

\noindent
Therefore, the law of total probability yields,
\begin{align*}
&\mathbb{P}\Big(\esssup_{Q \in \mathcal{Q}_{n-2}} \hspace{0.1cm} \big|\theta_{n-2, m, N}(Q) - \theta_{n-2, m}(Q)  \big| \ge \delta \Big)\\ 
&\le \mathbb{P}\left(\esssup_{Q \in \mathcal{Q}_{n-2}} \hspace{0.1cm}\Big|\frac{1}{N} \sum_{p = 1}^{N} e^m(X_{t_{n-2}}^{[p]}) V^m_{n-1}(X_{t_{n-1}}^{[p]}, Q)  - \mathbb{E}\big(e^m(X_{t_{n-2}}) V^m_{n-1}( X_{t_{n-1}}, Q) \big) \Big| \ge \frac{\delta - K_2 \cdot \varepsilon}{K_1}\right) + \mathbb{P}\Big(\big(\Omega_{n-2}^{\varepsilon}\big)^{c}\Big)\\
&\le \frac{D_{s, n-2, m}}{(\delta - K_2 \cdot \varepsilon)^s N^{\frac{s}{2}}} + \frac{U_{n-2, m}}{\varepsilon^s N^{\frac{s}{2}}},
\end{align*}
where the majoration for the first probability in the second-last line comes from Proposition \ref{prob_diff} and constant $D_{a, n-2, m}$ embeds constant $K_1$. The majoration of $\mathbb{P}\Big(\big(\Omega_{n-2}^{\varepsilon}\big)^{c}\Big)$ is straightforward using successively Markov inequality and Lemma \ref{lem_zyg}. Then, choosing $\varepsilon = \rho \delta$ for some $\rho > 0$ sufficiently small yields,
 $$\mathbb{P}\left(\esssup_{Q \in \mathcal{Q}_{n-2}} \hspace{0.1cm}\big|\theta_{n-2, m, N}(Q) - \theta_{n-2, m}(Q)  \big| \ge \delta \right) \le \frac{C_{s, n-2, m}}{\delta^s N^{\frac{s}{2}}} \le     \left\{
    \begin{array}{ll}
        \frac{C_{s, n-2, m}}{\delta ^{s} N^{s/2}} \hspace{0.9cm} \text{if} \hspace{0.2cm} \delta \in (0, 1],\\
        \frac{C_{s, n-2, m}}{\delta ^{s/2} N^{s/2}} \hspace{0.6cm} \text{else}
    \end{array}
\right..$$
for some positive constant $C_{a, n-2, m}$. Now let us assume that the proposition holds for $k+1$ and show that it also holds for $k$. For any $Q \in \mathcal{T}_{k+1}$, it follows from triangle inequality that,
\begin{align*}
\big|\theta_{k, m, N}(Q) - \theta_{k, m}(Q) \big| &\le \big|(A_{m, N}^{k})^{-1}\big| \cdot \Big|\frac{1}{N} \sum_{p = 1}^{N} e^m(X_{t_k}^{[p]}) \Big(V^{m, N}_{k+1}( X_{t_{k+1}}^{[p]}, Q) - V^{m}_{k+1}(X_{t_{k+1}}^{[p]}, Q)\Big) \Big|\\
&\quad  +\big|(A_{m, N}^{k})^{-1}\big| \cdot  \Big|\frac{1}{N} \sum_{p = 1}^{N} e^m(X_{t_k}^{[p]})V^{m}_{k+1}(X_{t_{k+1}}^{[p]}, Q) - \mathbb{E}\big(e^m(X_{t_{k}}) V^m_{k+1}(X_{t_{k+1}}, Q) \big)\Big|\\
&\quad +\Big|(A_{m}^{k})^{-1}(A_{m}^{k} - A_{m, N}^{k})(A_{m, N}^{k})^{-1} \cdot \mathbb{E}\big(e^m(X_{t_{k}}) V^m_{k+1}(X_{t_{k+1}}, Q)\big) \Big|\\
&\le \big|(A_{m, N}^{k})^{-1}\big| \cdot \frac{1}{N} \sum_{p = 1}^{N} \big|e^m(X_{t_k}^{[p]})\big| \cdot \Big|V^{m, N}_{k+1}(X_{t_{k+1}}^{[p]}, Q) - V^{m}_{k+1}(X_{t_{k+1}}^{[p]}, Q) \Big|\\
&\quad + \big|(A_{m, N}^{k})^{-1}\big| \cdot  \Big|\frac{1}{N} \sum_{p = 1}^{N} e^m(X_{t_k}^{[p]})V^{m}_{k+1}(X_{t_{k+1}}^{[p]}, Q) - \mathbb{E}\left(e^m(X_{t_{k}}) V^m_{k+1}(X_{t_{k+1}}, Q) \right)\Big|\\
&\quad + \Big|(A_{m}^{k})^{-1}(A_{m}^{k} - A_{m, N}^{k})(A_{m, N}^{k})^{-1} \cdot \mathbb{E}\big(e^m(X_{t_{k}}) V^m_{k+1}(X_{t_{k+1}}, Q) \big) \Big|.
\end{align*}
But for all $1 \le p \le N$, Cauchy-Schwartz inequality yields,
\begin{align*}
\Big|V_{k+1}^{m, N}(X_{t_{k+1}}^{[p]}, Q) - V_{k+1}^{m}(X_{t_{k+1}}^{[p]}, Q) \Big| &\le \esssup_{q \in \mathbb{A}_{k+1}(Q)} \hspace{0.1cm} \langle \theta_{k+1, m, N}(Q+q) - \theta_{k+1, m}(Q+q), e^m(X_{t_{k+1}}^{[p]}) \rangle\\
&\le \big|e^m(X_{t_{k+1}}^{[p]}) \big| \cdot \esssup_{q \in \mathbb{A}_{k+1}(Q)} \hspace{0.1cm} \big|\theta_{k+1, m, N}(Q+q) - \theta_{k+1, m}(Q+q)\big|.
\end{align*}
Thus,
\begin{align*}
\big|\theta_{k, m, N}(Q) - \theta_{k, m}(Q) \big| &\le \left(  \frac{\big|(A_{m, N}^{k})^{-1}\big|}{N} \sum_{p = 1}^{N} \big|e^m(X_{t_k}^{[p]})\big| \cdot 
\big|e^m(X_{t_{k+1}}^{[p]}) \big| \right) \esssup_{q \in \mathbb{A}_{k+1}(Q)}  \big|\theta_{k+1, m, N}(Q+q) - \theta_{k+1, m}(Q+q)\big|\\
&\quad + \big|(A_{m, N}^{k})^{-1}\big| \cdot  \Big|\frac{1}{N} \sum_{p = 1}^{N} e^m(X_{t_k}^{[p]})V^{m}_{k+1}(X_{t_{k+1}}^{[p]}, Q) - \mathbb{E}\left(e^m(X_{t_{k}}) V^m_{k+1}(X_{t_{k+1}}, Q) \right)\Big|\\
&\quad + \Big|(A_{m}^{k})^{-1}(A_{m}^{k} - A_{m, N}^{k})(A_{m, N}^{k})^{-1} \cdot \mathbb{E}\big(e^m(X_{t_{k}}) V^m_{k+1}(X_{t_{k+1}}, Q) \big) \Big|.
\end{align*}
Therefore, on $\Omega_{k}^{\varepsilon}$, there exists some positive constants $K_1, K_2, K_3$ such that,
\begin{align*}
& \esssup_{Q \in \mathcal{Q}_{k}} \hspace{0.1cm} \big|\theta_{k, m, N}(Q) - \theta_{k, m}(Q) \big| \\
&\le K_1\underbrace{\left(\frac{1}{N} \sum_{p = 1}^{N} \big|e^m(X_{t_k}^{[p]})\big| \cdot \big|e^m(X_{t_{k+1}}^{[p]}) \big|  \right)}_{I_N^1} \cdot 
\underbrace{\esssup_{Q \in \mathcal{Q}_{k+1}} \hspace{0.1cm} \big|\theta_{k+1, m, N}(Q) - \theta_{k+1, m}(Q)\big|}_{I_N^2}\\
&\quad + K_2 \cdot \underbrace{ \esssup_{Q \in \mathcal{Q}_{k+1}} \hspace{0.1cm} \Big|\frac{1}{N} \sum_{p = 1}^{N} e^m(X_{t_k}^{[p]})V^{m}_{k+1}(X_{t_{k+1}}^{[p]}, Q) - \mathbb{E}\left(e^m(X_{t_{k}}) V^m_{k+1}(X_{t_{k+1}}, Q) \right)\Big|}_{I_N^3}+ K_3 \cdot \varepsilon,
\end{align*}
where to obtain the coefficient $K_3$ in the last inequality, we used the fact that,
$$\esssup_{Q \in \mathcal{Q}_{k+1}} \hspace{0.1cm} \mathbb{E}\left(e^m(X_{t_{k}}) V^m_{k+1}(X_{t_{k+1}}, Q) \right) \le \underset{Q \in \mathcal{T}_{k+1}}{\sup} \hspace{0.1cm} \mathbb{E}\left(e^m(X_{t_{k}}) V^m_{k+1}(X_{t_{k+1}}, Q) \right) < + \infty.$$
The term $I_N^3$ can be handled using Proposition \ref{prob_diff}. Then, it suffices to prove that,
$$\mathbb{P}\big( I_N^1 \cdot I_N^2 \ge \delta \big) \le \frac{K}{\delta^a \cdot N^{a/2}}$$
for some positive constant $K$. But we have,
\begin{equation}
\label{bonf_ine}
\mathbb{P}\big( I_N^1 \cdot I_N^2\ge \delta \big) = 1 - \mathbb{P}\big( I_N^1 \cdot I_N^2 \le \delta \big) \le 1 - \mathbb{P}\big( I_N^1  \le \sqrt{\delta}; I_N^2 \le \sqrt{\delta} \big) \le  \mathbb{P}\big( I_N^1 \ge \sqrt{\delta} \big) + \mathbb{P}\big(I_N^2 \ge \sqrt{\delta} \big).
\end{equation}
Moreover, by the induction assumption, we know that, there exists a positive constant $B_{a, k, m}$ such that,
$$\mathbb{P}\big(I_N^2 \ge \sqrt{\delta} \big) \le \frac{B_{s, k, m}}{\delta ^{s/2} N^{s/2}} \le     \left\{
    \begin{array}{ll}
        \frac{B_{s, k, m}}{\delta ^{s} N^{s/2}} \hspace{0.9cm} \text{if} \hspace{0.2cm} \delta \in (0, 1],\\
        \frac{B_{s, k, m}}{\delta ^{s/2} N^{s/2}} \hspace{0.6cm} \text{otherwise.}
    \end{array}
\right.$$
In addition, it follows from Markov inequality and Lemma \ref{lem_zyg}  that there exists a positive constant $M_{a, k, m}$ such that
$$\mathbb{P}\big( I_N^1 \ge \sqrt{\delta} \big) \le \frac{M_{s, k, m}}{\delta^s N^{s/2}} \le     \left\{
    \begin{array}{ll}
        \frac{M_{s, k, m}}{\delta ^{s} N^{s/2}} \hspace{0.9cm} \text{if} \hspace{0.2cm} \delta \in (0, 1],\\
        \frac{M_{s, k, m}}{\delta ^{s/2} N^{s/2}} \hspace{0.6cm} \text{otherwise.}
    \end{array}
\right.$$
Hence, there exists a positive constant $C_{s, k, m}$ such that,
$$\mathbb{P}\big( I_N^1 \cdot I_N^2 \ge \delta \big) \le  \left\{
    \begin{array}{ll}
        \frac{C_{s, k, m}}{\delta ^{s} N^{s/2}} \hspace{0.9cm} \text{if} \hspace{0.2cm} \delta \in (0, 1],\\
        \frac{C_{s, k, m}}{\delta ^{s/2} N^{s/2}} \hspace{0.6cm} \text{otherwise}
    \end{array}
\right.$$

\noindent
and this completes the proof.
\end{proof}

We now state the last result of this paper concerning a deviation inequality involving the actual swing value function.

\begin{Proposition}
Consider assumptions $\bm{\mathcal{H}_{3, \infty}}$ and $\bm{\mathcal{H}_{4, \infty}^{LS}}$. For all $k=0, \ldots, n-2$, $\delta > 0$ and $s \ge 2$ there exists a positive constant $C_{s, k, m}$ such that,
$$\mathbb{P}\left(\esssup_{Q \in \mathcal{Q}_k} \hspace{0.1cm} \Big|V_k^{m, N}\big(X_{t_k}, Q\big) - V_k^{m}\big(X_{t_k}, Q\big)   \Big| \ge \delta \right) \le \frac{C_{s, k, m}}{b(s, \delta) \cdot N^\frac{s}{2}}.$$
\end{Proposition}

\begin{proof}
Using the inequality, $|\underset{i \in I}{\sup} \hspace{0.1cm} a_i - \underset{i \in I}{\sup} \hspace{0.1cm} b_i| \hspace{0.1cm} \le \hspace{0.1cm} \underset{i \in I}{\sup} \hspace{0.1cm} |a_i-b_i|$ and then Cauchy-Schwartz' inequality, we have,
$$\esssup_{Q \in \mathcal{Q}_k} \hspace{0.1cm} \Big| V_k^{m, N}\big(X_{t_k}, Q\big) - V_k^{m}\big(X_{t_k}, Q\big) \Big| \le \big|e^m(X_{t_k}) \big| \cdot \esssup_{Q \in \mathcal{Q}_{k+1}} \hspace{0.1cm} \Big| \theta_{k + 1, m, N}(Q) - \theta_{k+1, m}(Q) \Big|.$$

\noindent
Thus, using the same argument as in \eqref{bonf_ine}, we get,
\begin{align*}
& \mathbb{P}\left( \esssup_{Q \in \mathcal{Q}_k} \hspace{0.1cm} \Big| V_k^{m, N}\big(X_{t_k}, Q\big) - V_k^{m}\big(X_{t_k}, Q\big) \Big| \ge \delta \right)\\  
&\le \mathbb{P}\left( \big|e^m(X_{t_k}) \big| \cdot \esssup_{Q \in \mathcal{Q}_{k+1}} \hspace{0.1cm} \Big| \theta_{k + 1, m, N}(Q) - \theta_{k+1, m}(Q) \Big| \ge \delta \right)\\
&\le \mathbb{P}\Big( \big|e^m(X_{t_k}) \big| \ge \sqrt{\delta} \Big) + \mathbb{P}\left( \esssup_{Q \in \mathcal{Q}_{k+1}} \hspace{0.1cm} \Big| \theta_{k + 1, m, N}(Q) - \theta_{k+1, m}(Q) \Big| \ge \sqrt{\delta} \right)\\
&\le \frac{K_{s, k, m}^1}{\delta^{s/2} \cdot N^{s/2}} + \frac{K_{s, k, m}^2}{b(s, \delta) \cdot N^{s/2}} \le \left\{
    \begin{array}{ll}
        \frac{K_{s, k, m}}{\delta ^{s} N^{s/2}} \hspace{0.9cm} \text{if} \hspace{0.2cm} \delta \in (0, 1]\\
        \frac{K_{s, k, m}}{\delta ^{s/2} N^{s/2}} \hspace{0.6cm} \text{otherwise}
    \end{array}
\right.
\end{align*}

\noindent
for some positive constant $K_{s, k, m}$, where the constant $K_{s, k, m}^1$ comes from Markov inequality (enabled by assumption $\bm{\mathcal{H}_{4, \infty}^{LS}}$). The existence of the positive constant $K_{s, k, m}^2$ results from Proposition \ref{tight_coeff_reg} (enabled by assumptions $\bm{\mathcal{H}_{3, \infty}}$ and $\bm{\mathcal{H}_{4, \infty}^{LS}}$). The coefficient $b(a, \delta)$ is also defined in Proposition \ref{tight_coeff_reg}. This completes the proof.
\end{proof}

\begin{remark}
The preceding proposition entails the following result as a straightforward corollary. For all $k = 0, \ldots, n-1$ and for any $Q \in \mathcal{T}_{k}$, we have,
\begin{equation*}
\mathbb{P}\left( \big|V_k^{m, N}\big(X_{t_k}, Q\big) - V_k^{m}\big(X_{t_k}, Q\big)   \big| \ge \delta \right) \le \frac{C_{s, k, m}}{b(s, \delta) \cdot N^\frac{s}{2}}.
\end{equation*}

\noindent
If we assume that $\underset{m \ge 1}{\sup} \hspace{0.1cm} C_{s, k, m} <  +\infty$, then for any $s \ge 2$, we have the following uniform convergence,
\begin{equation}
\label{prob_prem_diff}
\lim\limits_{N \rightarrow +\infty} \hspace{0.1cm} \underset{m \ge 1}{\sup} \hspace{0.1cm} \underset{Q \in \mathcal{T}_{k}}{\sup} \hspace{0.1cm} \mathbb{P}\left( \big|V_k^{m, N}\big(X_{t_k}, Q\big) - V_k^{m}\big(X_{t_k}, Q\big)   \big| \ge \delta \right) = 0.
\end{equation}

\noindent
But it follows from triangle inequality that,
\begin{align*}
& \mathbb{P}\left( \big|V_k^{m, N}\big(X_{t_k}, Q\big) - V_k\big(X_{t_k}, Q\big)   \big| \ge \delta \right)\\
&= 1 - \mathbb{P}\left( \big|V_k^{m, N}\big(X_{t_k}, Q\big) - V_k\big(X_{t_k}, Q\big)   \big| \le \delta \right)\\
&\le 1 - \mathbb{P}\left(\Big\{ \big|V_k^{m, N}(X_{t_k}, Q) - V_k^{m}(X_{t_k}, Q)   \big| \le \delta/2 \Big\} \cap \Big\{ \big|V_k^{m}(X_{t_k}, Q) - V_k(X_{t_k}, Q)   \big| \le \delta/2\Big\}\right)\\
&\le \mathbb{P}\Big( \big|V_k^{m, N}(X_{t_k}, Q) - V_k^{m}(X_{t_k}, Q)   \big| \ge \delta/2 \Big) + \mathbb{P}\Big( \big|V_k^{m}(X_{t_k}, Q) - V_k(X_{t_k}, Q)   \big| \ge \delta/2\Big)\\
&\le \mathbb{P}\Big( \big|V_k^{m, N}(X_{t_k}, Q) - V_k^{m}(X_{t_k}, Q)   \big| \ge \delta/2 \Big) + 4 \cdot \frac{\big\|V_k^{m}(X_{t_k}, Q) - V_k(X_{t_k}, Q) \big\|_2^2}{\delta^2},
\end{align*}

\noindent
where in the last inequality, we used Markov inequality. Then using Proposition \ref{cvg_m_basis} and result \eqref{prob_prem_diff} yields,
\begin{equation*}
\lim\limits_{m \rightarrow +\infty} \hspace{0.1cm} \lim\limits_{N \rightarrow +\infty} \underset{Q \in \mathcal{T}_k}{\sup} \hspace{0.1cm} \mathbb{P}\left( \big|V_k^{m, N}\big(X_{t_k}, Q\big) - V_k\big(X_{t_k}, Q\big)   \big| \ge \delta \right) = 0.
\end{equation*}

\noindent
The latter result implies that for a well-chosen and sufficiently large regression basis, the limit,
$$\lim\limits_{N \rightarrow +\infty} \underset{Q \in \mathcal{T}_k}{\sup} \hspace{0.1cm} \mathbb{P}\left( \big|V_k^{m, N}\big(X_{t_k}, Q\big) - V_k\big(X_{t_k}, Q\big)   \big| \ge \delta \right)$$

\noindent
may be arbitrary small insuring in some sense the theoretical effectiveness of the linear regression procedure in the context of swing pricing.
\end{remark}

\section*{Acknowledgments}
The author would like to thank Gilles Pagès and Vincent Lemaire for fruitful discussions. The author would also like to express his gratitude to Engie Global Markets for funding his PhD thesis.

\vspace{0.5cm}

\footnotesize
\noindent
\textbf{Funding}. The PhD thesis of the author is funded by the French ANRT (Association Nationale Recherche Technologie) and Engie Global Markets.

\subsection*{Declarations}

\textbf{Conflict of Interest}. The authors declare that they have no competing interests as defined by Springer, or other interests that might be perceived to influence the results and/or discussion reported in this paper.

\vspace{0.2cm}
\noindent
\textbf{Availability of data and material}. Not applicable.

\bibliographystyle{alpha}
\bibliography{biblio.bib}
%\nocite{*}

\appendix

\section{Appendix}

\subsection{Some useful results}
We present some materials used in this paper. The following lemma allows to show the continuity of the supremum of a continuous function when the supremum is taken over a set depending of the variable of interest.

\begin{lemma}
\label{conti_sup}
Consider a continuous function $f : \mathbb{R} \to \mathbb{R}$ and let $A$ and $B$ be two non-increasing and continuous real-valued functions defined on $\mathbb{R}$ such that for all $Q \in \mathbb{R}, A(Q) \le B(Q)$. Then the function
$$g: Q \in \mathbb{R} \mapsto \underset{q \in [A(Q), B(Q)]}{\sup} \hspace{0.1cm} f(q)$$

\noindent
is continuous.
\end{lemma}

\begin{proof}
To prove this lemma, we proceed by proving the function $g$ is both left and right continuous. Let us start with the right-continuity. Let $Q \in \mathbb{R}$ and $h$ a positive real number. Since $A$ and $B$ are non-increasing functions, two cases can be distinguished

\vspace{0.2cm}
\noindent
$\doubleunderline{A(Q + h) \le A(Q) \le B(Q + h) \le B(Q).}$

\vspace{0.2cm}
\noindent
Using the definition of $g$, we have,
\begin{equation}
\label{decompo_cont_g}
g(Q + h) = \max\left(\underset{q \in [A(Q + h), A(Q)]}{\sup} \hspace{0.1cm} f(q), \underset{q \in [A(Q), B(Q + h)]}{\sup} \hspace{0.1cm} f(q) \right).
\end{equation}

\noindent
Since $f$ is continuous on the compact set $[A(Q + h), A(Q)]$, it attains its maximum on a point $\alpha(Q, h) \in [A(Q + h), A(Q)]$. Owing to the squeeze theorem, the latter implies that $ \lim\limits_{h \rightarrow 0} \alpha(Q, h) = A(Q)$ since $A$ is a continuous function. Thus it follows from the continuity of $f$
$$\lim\limits_{\substack{h \rightarrow 0 \\ >}} \hspace{0.1cm}\underset{q \in [A(Q + h), A(Q)]}{\sup} \hspace{0.1cm} f(q) = \lim\limits_{\substack{h \rightarrow 0 \\ >}} \hspace{0.1cm} f(\alpha(Q, h)) = f(A(Q)).$$

\noindent
Moreover, since $B(Q + h) \le B(Q)$, we have $\underset{q \in [A(Q), B(Q + h)]}{\sup} \hspace{0.1cm} f(q) \le \underset{q \in [A(Q), B(Q)]}{\sup} \hspace{0.1cm} f(q) = g(Q)$. Thus by the continuity of the maximum function and taking the limit in \eqref{decompo_cont_g} yields 
\begin{align*}
\lim\limits_{\substack{h \rightarrow 0 \\ >}} \hspace{0.1cm} g(Q + h) \le \lim\limits_{\substack{h \rightarrow 0 \\ >}} \hspace{0.1cm} \max\left(\underset{q \in [A(Q + h), A(Q)]}{\sup} \hspace{0.1cm} f(q), g(Q)\right) &=  \max\left(\lim\limits_{\substack{h \rightarrow 0 \\ >}} \hspace{0.1cm} \underset{q \in [A(Q + h), A(Q)]}{\sup} \hspace{0.1cm} f(q), g(Q)\right).\\
&= \max\big(f(A(Q)), g(Q)\big) \le g(Q).
\end{align*}

\noindent
It remains to prove that $\lim\limits_{\substack{h \rightarrow 0 \\ >}} \hspace{0.1cm} g(Q + h) \ge g(Q)$ to get the right-continuity. But since $A(Q + h) \le A(Q)$
\begin{equation}
\label{eq_ref_cont_sup}
g(Q) \le \underset{q \in [A(Q + h), B(Q)]}{\sup} \hspace{0.1cm} f(q) = \max\Big(g(Q+h), \underset{q \in [B(Q + h), B(Q)]}{\sup} \hspace{0.1cm} f(q) \Big).
\end{equation}

\noindent
As above, using the continuity of $f$ on the compact set $[B(Q + h), B(Q)]$ yields 
$$\lim\limits_{\substack{h \rightarrow 0 \\ >}} \hspace{0.1cm}\underset{q \in [B(Q + h), B(Q)]}{\sup} \hspace{0.1cm} f(q) = f(B(Q)).$$

\noindent
Therefore taking the limit in \eqref{eq_ref_cont_sup} yields
\begin{align*}
    g(Q) \le \max\left(\lim\limits_{\substack{h \rightarrow 0 \\ >}} \hspace{0.1cm} g(Q+h), f(B(Q)) \right) = \max\left(\lim\limits_{\substack{h \rightarrow 0 \\ >}} \hspace{0.1cm} g(Q+h), \lim\limits_{\substack{h \rightarrow 0 \\ >}} \hspace{0.1cm} f(B(Q + h)) \right) \le \lim\limits_{\substack{h \rightarrow 0 \\ >}} \hspace{0.1cm} g(Q+h).
\end{align*}

\noindent
where in the last inequality we used the fact that, $f(B(Q + h)) \le g(Q+h)$. This gives the right-continuity in this first case. Let us consider the second case.

\vspace{0.2cm}
\noindent
$\doubleunderline{A(Q + h) \le B(Q + h) \le A(Q) \le B(Q)}$

\vspace{0.2cm}
\noindent
Since $B(Q+h) \le A(Q)$, it follows from the definition of $g$ that,
\begin{align}
\label{eq_annexe_cont_sup_1}
    \lim\limits_{\substack{h \rightarrow 0 \\ >}} \hspace{0.1cm} g(Q+h) \le  \max\left(\lim\limits_{\substack{h \rightarrow 0 \\ >}} \hspace{0.1cm} \underset{q \in [A(Q + h), A(Q)]}{\sup} \hspace{0.1cm} f(q), g(Q) \right) = \max\big(f(A(Q)), g(Q) \big) = g(Q).
\end{align}

\noindent
where we used as above the continuity of $f$ on the compact set $[A(Q + h), A(Q)]$. Moreover, notice that
$$g(Q) \le \underset{q \in [A(Q + h), B(Q)]}{\sup} f(q) = \max \Big(g(Q+h), \underset{q \in [B(Q + h), B(Q)]}{\sup} \hspace{0.1cm} f(q) \Big)$$

\noindent
Then, taking the limit in the last inequality yields,
\begin{align}
\label{eq_annexe_cont_sup_2}
    g(Q) \le \max\Big(\lim\limits_{\substack{h \rightarrow 0 \\ >}} \hspace{0.1cm} g(Q+h), \lim\limits_{\substack{h \rightarrow 0 \\ >}} \hspace{0.1cm} \underset{q \in [B(Q + h), B(Q)]}{\sup} \hspace{0.1cm} f(q) \Big) \nonumber &=\max\Big(\lim\limits_{\substack{h \rightarrow 0 \\ >}} \hspace{0.1cm} g(Q+h), f(B(Q)) \Big) \nonumber\\
    &= \max\Big(\lim\limits_{\substack{h \rightarrow 0 \\ >}} \hspace{0.1cm} g(Q+h), \lim\limits_{\substack{h \rightarrow 0 \\ >}} \hspace{0.1cm} f(B(Q+h)) \Big) \nonumber\\
    &\le \lim\limits_{\substack{h \rightarrow 0 \\ >}} \hspace{0.1cm} g(Q+h).
\end{align}

\noindent
Thus, from equations \eqref{eq_annexe_cont_sup_1} and \eqref{eq_annexe_cont_sup_2} one may deduce that $\lim\limits_{\substack{h \rightarrow 0 \\ >}} \hspace{0.1cm} g(Q+h) = g(Q)$. So that $g$ is a right-continuous function. Proving the left-continuity can be handled in the same way. The idea is the following. We start with $h$ a negative real number and consider the two following cases: $A(Q) \le A(Q+h) \le B(Q) \le B(Q+h)$ and $A(Q) \le B(Q) \le A(Q+h) \le B(Q+h)$ and proceed as for the right-continuity. Which will give $\lim\limits_{\substack{h \rightarrow 0 \\ <}} \hspace{0.1cm} g(Q+h) = g(Q)$. Therefore $g$ is a continuous function on $\mathbb{R}$.
\end{proof}

The following theorem also concerns the continuity of function in a parametric optimization.

\begin{theorem}
\label{inf_comp}
If $X, Y$ are topological spaces and $Y$ is compact, then for any continuous function $f : X \times Y \to\mathbb{R}$, the function $g(x) := \underset{y \in Y}{\inf} f(x,y)$ is well-defined and continuous.
\end{theorem}

\begin{proof}
Note that $g(x) > -\infty$ since for any fixed $x\in X, f(x,\cdot):Y\to\mathbb{R}$ is a continuous function defined on a compact space, and hence the infimum is attained. Then using that the sets $(-\infty,a)$ and $(b,\infty)$ form a subbase for the topology of $\mathbb{R}$, it suffices to check that $g^{-1}((-\infty,a))$ and $g^{-1}((b,\infty))$ are open. Let $\pi_X$ be the canonical projection $\pi_X:X\times Y\to X$, which we recall is continuous and open. It is easy to see that $g^{-1}((-\infty,a)) = \pi_X \circ f^{-1}((-\infty,a))$. Thus since $f$ and $\pi_X$ are continuous, $g^{-1}((-\infty,a))$ is open.

\vspace{0.2cm}
We now need to show that $g^{-1}((b,\infty))$ is open. We rely on the compactness of $Y$. Observe that,
$$g(x) > b \implies f(x,y) > b~ \forall y \implies \forall y, (x,y) \in f^{-1}((b,\infty)).$$ 

\noindent
Since $f$ is continuous, then $f^{-1}((b,\infty))$ is open. The latter implies that for all $ x\in g^{-1}((b,\infty))$ and for all $y\in Y$ there exists a \q{box} neighborhood $U_{(x,y)}\times V_{(x,y)}$ contained in $f^{-1}((b,\infty))$. Now using compactness of $Y$, a finite subset $\{(x,y_i)\}$ of all these boxes cover $\{x\}\times Y$ and we get,
$$\displaystyle \{x\}\times Y \subset \left( \cap_{i = 1}^k U_{(x,y_i)}\right)\times Y \subset f^{-1}((b,\infty))$$

\noindent
and hence $\displaystyle g^{-1}((b,\infty)) = \cup_{x\in g^{-1}((b,\infty))} \cap_{i = 1}^{k(x)} U_{x,y_i}$ is open. Which completes the proof.
\end{proof}

\begin{Proposition}[Gram determinant]
\label{gram_det}
Let $F$ be a linear subspace with dimension $n$ of a pre-Hilbert space $E$. Consider $(x_1, \ldots, x_n)$ as a basis of $F$ and $x \in E$. Let $p(x)$ denotes the orthogonal projection of $x$ onto $F$. Then,
$$G(x, x_1, \ldots, x_n) = \big\| x - p(x) \big\|^2 \cdot G(x_1, \ldots, x_n)$$

\noindent
where $G(x_1, \ldots, x_n)$ denotes the Gram determinant associated to $(x_1, \ldots, x_n)$.
\end{Proposition}

\begin{proof}
Note that $p(x)$ is a linear combination of $(x_i)_{1 \le i \le n}$. Since the determinant is stable by elementary operation, we have
$$G(x, x_1, \ldots, x_n) = G\big(x - p(x), x_1, \ldots, x_n\big).$$

\noindent
But $x - p(x)$ is orthogonal to each $x_i$ so that,
$$G\big(x - p(x), x_1, \ldots, x_n\big) =  \big\| x - p(x) \big\|^2 \cdot G(x_1, \ldots, x_n).$$

\noindent
this completes the proof.
\end{proof}

\subsection{Correspondences}
\textit{This section concerns correspondence and the well known Berge's maximum theorem. For a thorough analysis of the concept of correspondence, one may refer to Chapter 2 and 6 in \cite{berge1997topological}.}
\label{corresp}

\begin{definition}[Correspondence]
  Let $X$ and $Y$ be two non-empty sets. 

\begin{itemize}
	\item a \textbf{correspondence} $\Gamma$ from $X$ to $2^Y$ (noted: $\Gamma: X \rightrightarrows 2^Y$) is a mapping that associates for all $x \in X$ a subset $\Gamma(x)$ of $Y$. Moreover for all subset $S \subseteq X$, $\Gamma(S) := \cup_{x \in S}^{}\hspace{0.1cm}\Gamma(x)$.
	
	\item a correspondence $\Gamma$ is \textbf{single-valued} if $Card(\Gamma(x)) = 1$ for all $x \in X$
	
	\item a correspondence $\Gamma$ is \textbf{compact-valued} (or \textbf{closed-valued}) if for all $x \in X$, $\Gamma(x)$ is a compact (or closed) set.
\end{itemize}  
  
\end{definition}

Notice that a single-valued correspondence can be thought of as a function mapping $X$ into $Y$. Thus as correspondences appear to be a generalization of functions 
some properties or definitions in functions has their extension in correspondences. Specially the continuity for a classic numerical function is a particular case of 
the hemicontinuity for a correspondence. We only present the sequential characterization.

\begin{Proposition}[Sequential characterization of hemicontinuity]
\label{seq_carac_hemicont}
Let $(X, d_X)$ and $(Y, d_Y)$ be two metric spaces and $\Gamma: X \rightrightarrows 2^Y$ a correspondence.

\begin{itemize} 
    \item $\Gamma$ is lower hemicontinuous at $x \in X$ if and only if for all sequence $(x_n)_{n \in \mathbb{N}} \in X^{\mathbb{N}}$ that converges towards $x$, for all $y \in \Gamma(x)$ there exists a subsequence $(x_{n_k})_{k \in \mathbb{N}}$ of $(x_n)_{n \in \mathbb{N}}$ and a sequence $(y_k)_{k \in \mathbb{N}}$ such that $y_k \in \Gamma(x_{n_k})$ for all $k \in \mathbb{N}$ and $y_k \to y$.
    \vspace{0.2cm}
    
    \item if $\Gamma$ is upper hemicontinuous at $x \in X$ then for all sequence $(x_n)_{n \in \mathbb{N}} \in X^{\mathbb{N}}$ and all sequence $(y_n)_{n \in \mathbb{N}}$ such that for all $n \in \mathbb{N}, y_n \in \Gamma(x_n)$, there exists a convergent subsequence of $(y_n)_{n \in \mathbb{N}}$ whose limit lies in $\Gamma(x)$. If $Y$ is compact then, the converse holds true.
\end{itemize}
\end{Proposition}

An important result relating correspondence and parametric optimization is the Berge's maximum theorem.

\begin{Proposition}[Berge's maximum theorem]
\label{max_th}
Let $\mathcal{Q}$ and $Y$ be two topological spaces, $\Gamma: \mathcal{Q} \rightrightarrows 2^Y$ a compact-valued and continuous correspondence and $\phi$ a continuous function on the product space $Y \times \mathcal{Q}$. Define for all $Q\in \mathcal{Q}$
$$\sigma(Q) := \argmax_{q \in \Gamma(Q)} \hspace{0.1cm} \phi(q, Q) \hspace{0.6cm} \phi^{*}(Q) := \underset{q \in \Gamma(Q)}{\max} \hspace{0.1cm} \phi(q, Q).$$

\noindent
Then,
\begin{itemize}
    \item The correspondence $\sigma: \mathcal{Q} \rightrightarrows Y$  is compact-valued, upper hemicontinuous, and closed.
    \item The function $\phi^{*}: \mathcal{Q} \to \mathbb{R}$ is continuous.
\end{itemize}
\end{Proposition}

\end{document}